\documentclass[prx,nofootinbib,onecolumn,showkeys,groupaddress,preprintnumbers,floatfix,
superscriptaddress]{revtex4-1}

\usepackage{etex}                        \usepackage{ifpdf}                       \usepackage{xspace}                      \usepackage[table]{xcolor}               \usepackage{setspace}                    

\PassOptionsToPackage{final}{graphicx}   \usepackage{graphicx}                    

\makeatletter                            \@ifclassloaded{beamer}{}{\PassOptionsToPackage{pagebackref}{hyperref}}
\makeatother
\usepackage{hyperref}
\usepackage{url}                         \usepackage[super]{nth}                  \usepackage{contour}                     \usepackage{verbatim}                    \usepackage{ragged2e}                    \usepackage{attrib}                      \usepackage{adjustbox}                   \usepackage[absolute,overlay]{textpos}   \usepackage{calc}                        

\usepackage{amsmath}                     \usepackage{amscd}                       \usepackage{amsfonts}                    \usepackage{amssymb}                     \usepackage{amsthm}                      \usepackage{scalerel}                    \usepackage{bm}                          \usepackage{bbm}                         \usepackage{braket}				               \usepackage{mathtools}                   \usepackage{etoolbox}                    \usepackage{textcomp}                    \usepackage{stmaryrd}                    \usepackage{nicefrac}                    \usepackage[binary-units]{siunitx}       

\usepackage[T1]{fontenc}                 \usepackage{lmodern}                     \usepackage[scaled=0.72]{beramono}       \usepackage{microtype}                   

\usepackage{algorithm}                   \usepackage{algorithmicx}                

\usepackage{booktabs}                    \usepackage{dcolumn}                     

\makeatletter
  \@ifclassloaded{revtex4-1}
  {
    \usepackage[caption=false]{subfig}   }
  {
    \usepackage{caption}
    \usepackage{subcaption}
  }
  \@ifclassloaded{beamer}
  {
    \usepackage[backend=biber,
                autocite=superscript,
                doi=false,
                url=false,
                sorting=none]{biblatex}
    \usepackage[useregional,showdow]{datetime2}
  }
  {
\usepackage{enumitem}                }
\makeatother

\definecolor{ryan}{RGB}{64, 0, 64}
\definecolor{nix}{RGB}{255, 0, 0}

\definecolor{ucdblue1}{cmyk}{.87,.46,0,.49} \definecolor{ucdblue2}{cmyk}{1., .56, 0., .34}
\colorlet{ucdblue}{ucdblue2}

\makeatletter
\def\@lox@prtc{\section*{\@fxlistfixmename}\begingroup\def\@dotsep{4.5}}
\def\@lox@psttc{\endgroup}
\makeatother

\usepackage{epigraph}

\colorlet {past_color}    {red}
\colorlet {pres_color}    {blue}
\colorlet {futu_color}    {black!30!green}

\colorlet {temp_color_1}  {red!50!blue}
\colorlet {temp_color_2}  {red!50!green}
\colorlet {temp_color_3}  {blue!50!green}

\colorlet {hmu_color}     {blue!67!green}
\colorlet {rhomu_color}   {temp_color_1!80!blue}
\colorlet {rmu_color}     {blue}
\colorlet {bmu_1_color}   {temp_color_1}
\colorlet {bmu_2_color}   {temp_color_3}
\colorlet {qmu_color}     {temp_color_1!67!green}
\colorlet {wmu_color}     {temp_color_2!57!blue}
\colorlet {sigmamu_color} {temp_color_2}

\usepackage{listings}
\lstdefinestyle{mypython}{
language=Python,                        basicstyle=\small\ttfamily,             keywordstyle=\color{green!50!black},    commentstyle=\color{gray},              numbers=left,                           numberstyle=\tiny,                      stepnumber=1,                           numbersep=5pt,                          backgroundcolor=\color{gray!10},        frame=none,                             tabsize=2,                              captionpos=b,                           breaklines=true,                        breakatwhitespace=false,                showspaces=false,                       showtabs=false,                         morekeywords={as},                      }

\theoremstyle{plain}    
\theoremstyle{plain}    
\theoremstyle{plain}    
\theoremstyle{plain}    
\theoremstyle{plain}    
\theoremstyle{plain}    
\theoremstyle{plain}    
\theoremstyle{plain}    
\theoremstyle{plain}    
\theoremstyle{plain}    
\theoremstyle{plain}    
\theoremstyle{plain}

\newcommand{\CausalState}       { \mathcal{S} }

\newcommand{\forward}{+}
\newcommand{\reverse}{-}
\newcommand{\forwardreverse}{\pm} 

\newcommand{\FutureCausalState} { {\CausalState}^{\forward} }

\newcommand{\PastCausalState}   { {\CausalState}^{\reverse} }

\newcommand{\lastindex}[2]{
  \edef\tempa{0}
  \edef\tempb{#2}
  \ifx\tempa\tempb
\edef\tempc{#1}
  \else
\edef\tempa{0}
    \edef\tempb{#1}
    \ifx\tempa\tempb
      \edef\tempc{#2}
    \else
      \edef\tempc{#1+#2}
    \fi
  \fi
  \tempc
}

\newcommand{\CSjoint}[1][,]{
   \edef\tempa{:}
   \edef\tempb{#1}
   \ifx\tempa\tempb
\ensuremath{\FutureCausalState\!#1\PastCausalState}
   \else
\ensuremath{\FutureCausalState#1\PastCausalState}
   \fi
}
\newcommand{\CSjointKL}[3][,]{
   \edef\tempa{:}
   \edef\tempb{#1}
   \ifx\tempa\tempb
\ensuremath{\FutureCausalState_{#2}\!#1\PastCausalState_{#3}}
   \else
\ensuremath{\FutureCausalState_{#2}#1\PastCausalState_{#3}}
   \fi
}

\newif\ifpm
\edef\tempa{\forwardreverse}
\edef\tempb{\pm}
\ifx\tempa\tempb
   \pmfalse
\else
   \pmtrue
\fi

\ifcsname mathclap\endcsname
  \relax
\else
  \def\clap#1{\hbox to 0pt{\hss#1\hss}}

\fi

\newcommand{\op} [3] [] {
  \ensuremath{
    \operatorname{#2_{#1}}
    \if\relax\detokenize{#3}\relax
    \else
      \left[ #3 \right]
    \fi
  }
  \xspace
}

\makeatletter
  \@ifclassloaded{beamer}
  {
    \hypersetup{
      linkcolor={blue!50!black},
      citecolor={blue!50!black},
      urlcolor={blue!80!black},
    }
  }
  {
    \hypersetup{
      colorlinks,
      linkcolor={blue!50!black},
      citecolor={blue!50!black},
      urlcolor={blue!80!black},
    }
  }
\makeatother

\usepackage{empheq}
\usepackage{overpic}
\usepackage{multirow}

\usepackage{fancybox}
\usepackage{framed}
\definecolor{shadecolor}{rgb}{0.95,0.95,0.9}

\newcommand{\argmin}{\text{argmin}}
\newcommand{\argmax}{\text{argmax}}

\newcommand{\kB}{k_\text{B}}

\newcommand{\stationary}{\boldsymbol{\pi}}

\newcommand{\EP}{\Sigma}
\newcommand{\EF}{\Phi}
\newcommand{\tr}{\text{tr}}

\newcommand{\EFO}{\boldsymbol{\Phi}}

\newcommand{\genop}{\mathcal{X}}
\newcommand{\sysref}{I/d}  
\newcommand{\constname}{\eta}
\newcommand{\const}{\tfrac{d-1}{d}}
 
\newcommand{\ident}{I}

\newcommand{\pp}{\vec{\pi}}
\newcommand{\bb}{\vec{b}}
\newcommand{\pbb}[1][\bb]{\tilde{\rho}(#1)}
\newcommand{\pbbF}[1][\bb]{\tilde{\rho}'(#1)}

\newcommand{\Hess}{\mathbf{H}}

\newtheorem{thm}{Theorem}

\newcommand{\barrier}{g}
\newcommand{\tilt}{f}

\makeatletter
\@booleanfalse\titlepage@sw
\makeatother

\begin{document}

\def\ourTitle{Thermodynamically ideal 
quantum-state
inputs to 
any device
}

\def\ourAbstract{We investigate and 
ascertain the ideal inputs to any finite-time thermodynamic process.  
We demonstrate that the expectation values of
entropy flow, heat, and work can all be determined via Hermitian observables of  
the initial state.  
These Hermitian operators 
encapsulate the breadth of 
behavior and the ideal inputs for common thermodynamic objectives. 
We show how to construct these Hermitian operators 
from measurements of thermodynamic output from a finite number of effectively arbitrary inputs.
Behavior of a small number of test inputs
thus determines the full range of 
thermodynamic behavior from all inputs.
For any process,
entropy flow, heat, and work can all be extremized by pure input states---eigenstates 
of the respective operators.
In contrast, the input states that minimize entropy production or maximize the change in free energy are non-pure mixed states obtained 
from the operators
as the solution of a convex optimization problem.
To attain these, 
we provide an easily implementable 
gradient descent method on the manifold of density matrices, where an analytic solution yields 
a valid direction of descent at each iterative step.
Ideal inputs within a limited domain, and their associated thermodynamic operators, are obtained with less effort.
This allows 
analysis of
ideal thermodynamic inputs within quantum subspaces of infinite-dimensional quantum systems; 
it also allows 
analysis of
ideal inputs in 
the classical limit.
Our examples illustrate the diversity of `ideal' inputs:
Distinct initial states minimize entropy production, extremize the change in free energy, and maximize work extraction.
}

\def\ourKeywords{nonequilibrium thermodynamics, 
  quantum thermodynamics,
  entropy production, 
  work extraction, generalized Bloch vectors
}

\hypersetup{
  pdfauthor={Paul M. Riechers},
  pdftitle={\ourTitle},
  pdfsubject={\ourAbstract},
  pdfkeywords={\ourKeywords},
  pdfproducer={},
  pdfcreator={}
}

\title{\ourTitle}

\author{Paul M. Riechers}
\email{pmriechers@gmail.com}

\affiliation{Beyond Institute for Theoretical Science, San Francisco, California, USA}

\author{Chaitanya Gupta}

\affiliation{Department of Applied Mathematics and Theoretical Physics,
	University of Cambridge, UK}

\author{Artemy Kolchinsky}

\affiliation{Universal Biology Institute, The University of Tokyo, 7-3-1 Hongo, Bunkyo-ku, Tokyo 113-0033, Japan}

\author{Mile Gu}
\email{mgu@quantumcomplexity.org}

\affiliation{Nanyang Quantum Hub,
	School of Physical and Mathematical Sciences,
Nanyang Technological University, Singapore}

\affiliation{CNRS-UNS-NUS-NTU International Joint Research Unit, UMI 3654, Singapore 117543, Singapore}

\affiliation{Centre for Quantum Technologies, National University of Singapore}

\date{\today}
\bibliographystyle{unsrt}

\begin{abstract}
\ourAbstract
\end{abstract}

\date{\today}
\maketitle

\setstretch{1.1}

\section{Introduction}

Throughout its history, thermodynamics primarily investigated
the efficiency of
various control processes for implementing a desired functionality.  
However, the complementary question of \emph{which initial physical states produce the best thermodynamic behavior} 
remains 
relatively unexplored.
Indeed, there is a historical reason for this:
In equilibrium transformations, the system always stays infinitesimally close to equilibrium, so there is no sense in asking about alternative inputs to the process.  
Yet
modern devices 
transform quantum and classical system rapidly.
These finite-time nonequilibrium transformations
have highly non-trivial initial-state dependence.
Here we explore the ideal thermodynamic inputs to such devices, where the system can be arbitrarily far from equilibrium throughout the transformation.

The initial-state dependence of entropy production and associated thermodynamic quantities has been explored only recently in relation to the ideal inputs, via mismatch costs~\cite{Kolc17, Riec21_Initial, Riec21_Impossibility, Kolc21_State}.
However, the minimally dissipative input was only characterized in the case of reset processes~\cite{Riec21_Impossibility, Kolc21_State} and, even then, a construction was only given for qubits~\cite{Riec21_Impossibility}.
In the following, we constructively identify the thermodynamically ideal inputs for a much broader class of objectives, including heat minimization, maximizing work extraction, and maximizing gain in free energy.  Moreover, the ideal inputs are characterized and constructively identified for systems of arbitrary finite dimensions, for any finite-time process.
The results apply to both quantum and classical systems.
The results thus apply broadly,
from biology to electronics to cosmology, 
wherever finite-time thermodynamics is relevant.

\section{Overview of Framework}

It is often desirable for a physical device to implement a fixed transformation on arbitrary input. 
For example, in quantum technology,
we often want to construct a device that implements 
a prescribed completely positive and trace preserving (CPTP) map.
The device can achieve this 
via a time-dependent protocol that partially controls the system's Hamiltonian and its interactions with its environment.
We assume 
that the 
control protocol and the initial state of the environment are fixed,
while we are free to select 
the initial (mixed or pure) state of the system to input to our device.
The initial joint state of the system--environment supersystem is then $\rho_0^\text{tot} = \rho_0 \otimes \rho_0^\text{env}$.
The joint system--baths supersystem evolves unitarily via $U_t$, such that the joint state at any later time $t$ is given as 
$\rho_t^\text{tot} = U_{t} \rho_0^\text{tot} U_{t}^\dagger$.
The reduced states of the system and environment are given at any time by the appropriate partial trace of the joint state:
$\rho_t = \tr_\text{env}(\rho_t^\text{tot} )$ 
and
$\rho_t^\text{env} = \tr_\text{sys}(\rho_t^\text{tot} )$.	
Any CPTP
transformation of the system can be achieved this way.

In the following sections, we will discover the thermodynamically ideal  inputs to any such CPTP transformation.
To achieve this, we first introduce \emph{thermodynamic operators} in Sec.~\ref{sec:ThermoOperators},
\emph{generalized Bloch vectors} in Sec.~\ref{sec:GenBlochVectors},
and \emph{thermodynamic vectors}
in Sec.~\ref{sec:ThermoVectors}.
We then
show how to construct
these thermodynamic vectors and operators
from experimental observations
in Sec.~\ref{sec:Construct}.
This finally allows us to discuss and construct the thermodynamically ideal inputs
in Secs.~\ref{sec:IdealType1Inputs}
and \ref{sec:IdealType2Inputs}.
Sec.~\ref{sec:Ideal_restriction}
shows how these results apply to restricted subspaces---which is important for understanding the classical limit, and for applying the results to low-energy subspaces of infinite-dimensional systems.
Finally, the examples in Secs.~\ref{sec:Examples} and \ref{sec:Example2}
illustrate physical implications of our results.

\section{Thermodynamic operators}
\label{sec:ThermoOperators}

Thermodynamic quantities like work, heat, and entropy flow are 
notoriously
path-dependent quantities.  
Even their average values depend on the time-dependent control protocol
and the time-dependent density matrices of system and environment.
For example,
the expectation value of
entropy flow can very generally be calculated as
\begin{align}
	\braket{\EF}_{\rho_0} 
	&= 
	- \kB \int_0^\tau \tr \bigl( \dot{\rho}_t^\text{env}  \ln \stationary_{t}^\text{env} \bigr) \, dt ~,
	\label{eq:EFdef}
\end{align}	
where $\stationary_{t}^\text{env}$ is a tensor product of local-equilibrium reference states for the environment at time $t$.
Familiar expressions like 
$\braket{\EF}_{\rho_0} = \int \tfrac{\braket{ \delta Q^{(b)} } }{T_t^{(b)}} \, dt$
and
$\braket{\EF}_{\rho_0} = \tfrac{\braket{Q^{(b)}} }{T_0^{(b)} } + \tfrac{\braket{Q^{(b')}} }{ T_0^{(b')} }$
are special cases of Eq.~\eqref{eq:EFdef},
where 
$Q^{(b)}$ is the change in energy of bath $b$,
and
$T_t^{(b)}$ is the temperature of bath $b$ at time $t$.
With $H_t$ as the time-dependent Hamiltonian of the system,
similar integral expressions for work 
\begin{align}
\braket{W}_{\rho_0} 
= 
\int_0^\tau \tr \bigl( \rho_t \dot{H}_t  \bigr) \, dt
\end{align}	
and heat absorbed by the system
\begin{align}
\braket{Q}_{\rho_0} 
= 
\int_0^\tau \tr \bigl(  \dot{\rho}_t H_t  \bigr) \, dt
\end{align}	
are familiar centerpieces of the thermodynamic arsenal.  In this path-dependent spirit, it has been famously emphasized that ``work is not an observable''~\cite{Talk07_Fluctuation}.

It is therefore perhaps surprising that the expectation values of work, heat, and entropy flow can all be determined via Hermitian observables of the initial state.  As we show in 
App.~\ref{sec:LinearFunctionals},
since the above expectation values
are all linear functionals of the initial density matrix of the system,
they can be expressed as 
\begin{align}
	\braket{\EF}_{\rho_0} 
	= 
	\tr(\rho_0 \EFO) ~, \qquad
	\braket{W}_{\rho_0} 
	= 
	\tr(\rho_0 \mathcal{W}) ~,
	\qquad \text{and} \quad
	\braket{Q}_{\rho_0} 
	= 
	\tr(\rho_0 \mathcal{Q}) ~,
	\label{eq:ThermoOperatorEquations}
\end{align}	
where $\EFO$, $\mathcal{W}$, and $\mathcal{Q}$ are Hermitian operators,
which we refer to as the \emph{expected-entropy-flow operator}, \emph{expected-work operator}, and \emph{expected-heat operator} respectively.
We will show that, for a general finite-time thermodynamic process,
these \emph{thermodynamic operators} can be constructed via experimental observations from a finite number of arbitrary inputs.
The operators in turn reveal the special collection of initial states that 
minimize entropy flow, minimize heat, and maximize work extraction.
The operators also allow direct calculation
of thermodynamic expectation values from any initial state.
 
The expressions in Eq.~\eqref{eq:ThermoOperatorEquations} all clearly take the form
\begin{align}
	\braket{X}_{\rho_0} 
	&= 
	\tr(\rho_0 \genop) ~,
	\label{eq:GenopTraceFormula}
\end{align}	 
which we will study in general.
We will say that a thermodynamic quantity is `type-I' 
when 
its expectation value can be expressed as a linear functional of the initial state $\tr(\rho_0 \genop )$, as in Eq.~\eqref{eq:GenopTraceFormula}.
As an added benefit of studying this general formulation,
our methods can also be applied to reconstruct the Hamiltonian, the expected-change-of-energy operator, and infinitely many other linear operators that conform to this general linear structure for expectation values.
If the expectation value is always real-valued,
then
$\genop$ is guaranteed to be Hermitian (see Thm.\ 2.4.3 of Hassani \cite{Hass99a}).

It is known that invasive measurements 
can change the expected value of work and other thermodynamic variables~\cite{Pera17_No}.
We emphasize that different measurement schemes
must be treated as distinct quantum processes, since
each implies a unique sequence of
dynamic interventions.
Accordingly, each measurement scheme induces its own set of thermodynamic operators.
For example,
in App.~\ref{sec:MeasurementSchemes},
we address the two-point measurements (TPM) scheme,
and
give an explicit construction of the expected-TPM-work operator $\mathcal{W}_\text{TPM}$,
for which
\begin{align}
\braket{W_\text{TPM}}_{\rho_0} = \tr(\rho_0 \mathcal{W}_\text{TPM} )~.
\end{align}
Indeed,
for any measurement scheme---TPM~\cite{Espo09_Nonequilibrium}, one-point measurement~\cite{Deff16_Quantum, Beye20_Work, Sone20_Quantum}, or any other scheme---thermodynamic operators can be constructed, and our framework can be applied to identify 
both the breadth of behavior and
the ideal inputs within the scheme.

It is important to note that each of these thermodynamic operators $\mathcal{X}$
contain all information needed for expectation values of the relevant thermodynamic quantity $X$, but do not contain the information that would be needed for expectation values of functions of the thermodynamic quantity, like $X^2$ or $e^X$.  
As a point of nomenclature, we note that this is not about whether the operators are `observable', as has been suggested~\cite{Talk07_Fluctuation}, 
but is rather an elementary fact about expectation values---$\braket{X}$ does not determine $\braket{f(X)}$. 
The source of insufficiency is simple to see if we take work as an example:
The full work distribution of a process, whether classical or quantum, has support on a space much larger than the dimension $d$ of the system.  Accordingly, the $d$ eigenvalues of the thermodynamic operator cannot represent the full probability distribution of the thermodynamic quantity.~\footnote{Instantaneous thermodynamic quantities---like energy, position, or momentum---offer a familiar exception, since they describe a variable with the same dimension as the system.  Accordingly, their $d \times d$ operators can be used to calculate all moments of their representative quantity.  Similarly, sufficiently high-dimensional representations of other thermodynamic operators would allow their encapsulation of higher moments.}

To address a final nuance,
we note that
a random variable $X$ in the quantum domain 
may depend on
further specification of subensembles, 
since a density matrix can be decomposed in many ways~\cite{Alla05_Fluctuations}.
Nevertheless, 
as discussed in 
App.~\ref{sec:RVsForQuantumThermo},
all decompositions of the initial density matrix lead to the same expectation value 
for a fixed quantum process.
Accordingly,
our notation ``$\braket{X}_{\rho_0}$''
unambiguously refers to the expectation value 
of the relevant thermodynamic quantity $X$
for all possible decompositions
of the quantum state.

Despite these caveats,
the thermodynamic operators serve immense utility.
Their $d$
eigenstates 
represent the only $d$ features that influence the expected value 
of the thermodynamic variable.
Once inferred for a process,
the thermodynamic operators tell the full breadth of expected behavior from any input, and identify the unique pure states leading to extremal behavior.  Moreover, the performance of any of the infinitely many possible inputs can be calculated simply and directly
from the thermodynamic operator, 
without needing to run a new experiment each time.

\section{Generalized Bloch vector}
\label{sec:GenBlochVectors}

We have promised that thermodynamic operators 
are the key to identifying thermodynamically ideal inputs to any process.
It will be important then to be able to construct these operators.
Our construction will lean on generalized Bloch vectors and related thermodynamic vectors, introduced in this and the next section respectively.

It is well known that the state of a qubit $\rho_t$ can be expressed via its Bloch vector $\vec{a}_t$:
\begin{align}
	\rho_t = I/2 + \vec{a}_t \cdot \vec{\sigma} / 2 ~,
	\label{eq:BlochQubit}
\end{align}
where 
$\vec{\sigma} = (\sigma_x, \sigma_y, \sigma_z)$ is the vector of Pauli matrices.	
For a quantum system of arbitrary finite dimension---i.e., a qudit $\rho_t$ acting on
a $d$-dimensional vector space $\mathcal{V}_d$---we 
achieve something similar via a slight adaptation of
Ref.~\cite{Jako01}. We choose any complete basis 
$(I/d, \Gamma_1, \Gamma_2, \dots \Gamma_{d^2 - 1})$ for linear operators acting on $\mathcal{V}_d$,
such that the Hermitian operators $\Gamma_n$ are all traceless and mutually orthogonal, satisfying
\begin{subequations}
	\begin{align}
		\tr( \Gamma_n) &= 0 , \; \text{ and } 
	\label{eq:traceless}  \\	
	\tr(\Gamma_m \Gamma_n) &=  \constname \, \delta_{m,n} ~,
\label{eq:orthonormal}
\end{align}
\end{subequations}
where we will choose the normalizing constant to be $\constname = \tfrac{d-1}{d}$.
Any density matrix then has a unique decomposition in the 
operator basis 
$\vec{\Gamma} = (\Gamma_1, \Gamma_2, \dots \Gamma_{d^2 - 1})$, 
described by the 
\emph{generalized Bloch vector} $\vec{b}_t \in \mathbb{R}^{d^2 - 1}$
via 
\begin{align}
	\rho_t = \sysref + \vec{b}_t \cdot \vec{\Gamma} ~.
	\label{eq:GenBlochExpansion}
\end{align}
Since the magnitude of the
Bloch vector is 
$b_t = \sqrt{\frac{\tr(\rho_t^2) d - 1}{d-1}}$,
the density matrix represents a pure state iff
the magnitude of its corresponding Bloch vector is one.  For $d>2$, not all points in the Bloch ball correspond to physical states, but the set of all physical states is nevertheless a convex set---the convex hull of the pure states, which all lie on a $2(d-1)$-dimensional submanifold of the $(d^{2}-2)$-dimensional surface of the Bloch sphere~\cite{Jako01}.

For concreteness, 
we can 
choose the ordered operator basis to be a scaled ordering of 
generalized Gell-Mann matrices---the generators of SU($d$).  See App.~\ref{sec:GGM} for details. 
The standard Bloch vector is then recovered 
in the familiar two-dimensional case of a qubit,
where then 
$\vec{\Gamma} = \vec{\sigma}/2 = ( \sigma_x/2, \, \sigma_y/2, \, \sigma_z/2 )$ and $\constname = 1/2$.
Alternatively, 
App.~\ref{sec:Composite_bases}
shows how to construct 
valid 
composite operator bases.

\section{Thermodynamic vectors}
\label{sec:ThermoVectors}

Leveraging the general Bloch decomposition Eq.~\eqref{eq:GenBlochExpansion} of the initial state,
we find that we can express each expectation value in Eq.~\eqref{eq:ThermoOperatorEquations} as 
\begin{align}
	\braket{X}_{\rho_0} 
	&= 
		\braket{X}_{\sysref} + \vec{b}_0 \cdot \vec{x} ~,
		\label{eq:BlochDecompOfExpectationValues}
\end{align}	 
where $\vec{x} \in \mathbb{R}^{d^2 - 1}$ is the relevant \emph{thermodynamic vector}
\begin{align}
\vec{x} = \tr(\vec{\Gamma} \genop) ~.
\end{align}
In particular, the thermodynamic vector could be the 
\emph{entropy-flow vector} 
$\vec{\varphi} = \tr(\vec{\Gamma} \EFO)$,
the 
\emph{work vector}
$\vec{w} = \tr(\vec{\Gamma} \mathcal{W})$,
or the 
\emph{heat vector}
$\vec{q} = \tr(\vec{\Gamma} \mathcal{Q})$.

Conversely, the thermodynamic operators can be constructed from the thermodynamic vectors:
\begin{align}
	\genop = 
	\braket{X}_{\sysref} I +  \vec{x} \cdot \vec{\Gamma} / \constname ~.
	\label{eq:ThermoOpFromVec}
\end{align}	 
Using Eqs.~\eqref{eq:traceless} and \eqref{eq:orthonormal}, it is easy to verify that 
Eq.~\eqref{eq:ThermoOpFromVec} satisfies
$\tr(\rho_0 \genop) 
	= 
	\braket{X}_{\sysref} + \vec{b}_0 \cdot \vec{x}$.

In the next sections, we show how both $\braket{X}_{\sysref} $ and the thermodynamic vector $\vec{x}$ can be obtained linear algebraically from experimental measurements of thermodynamic output from a finite number of almost arbitrary inputs~\footnote{The choice of inputs is arbitrary, except that they should all be linearly independent.  However, this is not much of a restriction: almost any $d^2$ or fewer inputs chosen at random will by linearly independent.}.
Via  Eq.~\eqref{eq:ThermoOpFromVec}, this allows us to experimentally reconstruct the thermodynamic operators from observations of any process.

\section{Constructing thermodynamic vectors and operators from observations} \label{sec:Construct}

Suppose an experimentalist has an apparatus to transform the state of a finite-dimensional quantum system.
(We will address infinite-dimensional systems in Sec.~\ref{sec:Ideal_restriction}.)
This experimentalist
measures the expectation value of the thermodynamic random variable $X$
that results from each of $d^2$ 
linearly independent inputs to their device.
I.e.,
they 
record the average quantity
from each of the
initial states $(\rho_0^{(n)})_{n=1}^{d^2}$ with corresponding generalized Bloch vectors $(\vec{b}_0^{(n)})_{n=1}^{d^2}$.
Note that the generalized Bloch vectors can be obtained as
$\vec{b}_0^{(n)} = \tr(\rho_0^{(n)} \vec{\Gamma}) / \constname $.
From Eq.~\eqref{eq:BlochDecompOfExpectationValues}, we see that
\begin{align}
	\underbrace{
		\begin{bmatrix}
			1&
			\vec{b}_0^{(1)} \\
			1&
			\vec{b}_0^{(2)} \\
			\vdots&
			\vdots \\
			1&
			\vec{b}_0^{(d^2)} 
		\end{bmatrix}
	}_{\eqqcolon B}
\begin{bmatrix}
		\braket{X}_{\sysref}  \\
		\vec{x}
	\end{bmatrix}
	&=
	\begin{bmatrix}
		\braket{X}_{\rho_0^{(1)}} \\
		\braket{X}_{\rho_0^{(2)}}  \\
		\vdots \\
		\braket{X}_{\rho_0^{(d^2)}} 
	\end{bmatrix}
	\label{eq:QuditLinalg}
\end{align}

Notice that the $B$ matrix defined in Eq.~\eqref{eq:QuditLinalg}
is invertible since the $d^2$ initial states are all linearly independent.
Hence,
with the $d^2$ measurements in hand and some simple linear algebra,
we can find both 
(i)
the expectation value from the reference input $\braket{X}_{\sysref}$
and
(ii)
the input-independent thermodynamic vector $\vec{x}$:
\begin{align}
	\begin{bmatrix}
		\braket{X}_{\sysref}  \\
		\vec{x}
	\end{bmatrix}
	&=
	B^{-1}
	\begin{bmatrix}
		\braket{X}_{\rho_0^{(1)}} \\
		\braket{X}_{\rho_0^{(2)}}  \\
		\vdots \\
		\braket{X}_{\rho_0^{(d^2)}} 
	\end{bmatrix} ~.
\end{align}
$\braket{X}_{\sysref} $ and $\vec{x}$ can now be used in 
Eq.~\eqref{eq:ThermoOpFromVec} to construct
the thermodynamic operator of interest: 
$\genop = \braket{X}_{\sysref} I +  \vec{x} \cdot \vec{\Gamma} / \constname $.
This 
constitutes a type of `operator tomography',
which is distinct but reminiscent of both quantum-state tomography and process tomography.

We are now equipped to identify and construct
the thermodynamically ideal inputs to any device.

\section{Ideal inputs for type-I objectives}
\label{sec:IdealType1Inputs}

Recall
that a thermodynamic quantity is `type-I' 
when 
its expectation value can be expressed as a linear functional of the initial state $\tr(\rho_0 \genop )$.
It is often desirable to minimize or maximize
type-I quantities,
e.g., to minimize heat or maximize work extraction.

An immediate observation can be drawn from 
Eq.~\eqref{eq:GenopTraceFormula} for any finite-dimensional quantum system,
by considering the spectral decomposition of the bounded operator 
$\genop = \sum_{n =1}^d \lambda_n \ket{v_n} \! \bra{v_n}$,
where $\Lambda_{\genop} = (\lambda_n)_n$
is the tuple of $\genop$'s eigenvalues
with corresponding eigenstates $V_{\genop}  = (\ket{v_n})_n$.

\begin{thm}
	\label{thm:Type1ExtremaArePure}
There is always a pure-state input that extremizes the expectation value
of type-I quantities $\braket{X}_{\rho_0} = \tr(\rho_0 \genop)$.
This pure state corresponds to an eigenstate of the thermodynamic operator $\genop$
with extremal eigenvalue.
\end{thm}

For example:
there is always a pure-state input that minimizes heat;
there is always a pure-state input that maximizes heat;
there is always a pure-state input that minimizes entropy flow;
and so on.

The implications are important and somewhat surprising.
For instance, 
consider the work extracted $W_\text{extracted} = - W$
during a 
finite-time cyclic work-extraction protocol from a finite-dimensional quantum system.
Thm.~\ref{thm:Type1ExtremaArePure} asserts that  
maximal work can be extracted from the
pure-state input $\ket{w_\text{min}} \in \argmin_{\ket{v} \in V_{\mathcal{W}} } \bigl\{ \braket{ v | \mathcal{W} | v } \bigr\}$. \footnote{This initial state $\rho_0 = \ket{w_\text{min}} \bra{w_\text{min}}$
yields the minimal expected work, and thus the maximal expected work extraction value $-\braket{w_\text{min} | \mathcal{W} | w_\text{min}} = \max \{ -\lambda : \lambda \in \Lambda_\mathcal{W} \}$.}
When the minimal-eigenvalue eigenspace of $\mathcal{W}$ is degenerate, maximal work extraction can be achieved by both pure and mixed inputs.
In contrast, as we will see later, the state that minimizes entropy production is generically a mixed state.  Therefore,
\emph{pure-state inputs that maximize work extraction imply non-minimal entropy production.}
Similarly,
\emph{pure-state inputs that minimize heat  imply non-minimal entropy production.}

One broad lesson from this analysis
is that maximizing work extraction, minimizing heat, minimizing entropy production, etc., are all distinct concepts that should not be conflated.
There is a great diversity in optimality
among different thermodynamic goals.
This will be emphasized again in our first example in Sec.~\ref{sec:Examples}.

It is worth noting that the smallest and largest eigenvalues of each thermodynamic operator
demarcate the range of corresponding expectation values 
that can be achieved via
alternative inputs:
\begin{align}
	\braket{X}_{\rho_0} \in [\min(\Lambda_{\genop}) , \,  \max(\Lambda_{\genop})] ~.
\end{align}	
All values in this continuous range are achievable by some input.

\section{Ideal inputs for type-II objectives}
\label{sec:IdealType2Inputs}

Not all thermodynamic expectation values can be expressed like Eq.~\eqref{eq:GenopTraceFormula} as a linear functional of the initial state.
In particular, 
entropy production, 
reduction in nonequilibrium free energy, 
and change in entropy
all have expectation values that are nonlinear functions of the initial state.
However, in each of these three cases, the nonlinearity is of the same form, since it derives from the change
of von Neumann entropy $S(\rho) = -\tr(\rho \ln \rho)$.  Accordingly, the initial states that extremize this second class of thermodynamic quantities all share similar features.

The expectation value of entropy production, 
$\braket{\EP}_{\rho_0} = \braket{\EF}_{\rho_0} + \kB \Delta S(\rho_t)$,
plays a central role in nonequilibrium thermodynamics.
When the environment begins in local equilibrium, and is uncorrelated with both itself and the system, then the famous Second Law of thermodynamics is valid:
$\braket{\EP}_{\rho_0} \geq 0$.
In the appropriate circumstances,
entropy production alternatively can be expressed
as the work performed beyond the change in nonequilibrium free energy
$T \braket{\EP}_{\rho_0} = \braket{W}_{\rho_0} - \Delta \mathcal{F}_t$,
where $T$ is the initial temperature of the environment.

The expectation values of
(i) entropy production,
(ii) change in free energy, and (iii) change in entropy 
are each proportional to
\begin{align}
f_{\rho_0}^{(\genop)}
\coloneqq  \tr(\rho_0 \genop) + S(\rho_\tau) - S(\rho_0) 
\label{eq:type2}
\end{align}
for the appropriate linear operators $\genop$.
We will say that a thermodynamic quantity is `type-II' 
when 
its expectation value is proportional to 
$f_{\rho_0}^{(\genop)}$ for some linear operator $\genop$ and some positive constant of proportionality.
In particular,
\begin{align}
\braket{\EP}_{\rho_0} = \kB f_{\rho_0}^{(\EFO / \kB)}
\quad
\text{and} 
\quad
\Delta S_t = f_{\rho_0}^{(0)}
\quad
\text{and} 
\quad
-\Delta \mathcal{F}_t = \kB T f_{\rho_0}^{( -(\mathcal{Q} + \mathcal{W} ) / \kB T)} ~.
\end{align}

\subsection{General processes}
\label{sec:GradientDescentAlgo}

\subsubsection{Inputs that minimize type-II functions}

\begin{thm}
	\label{thm:Type2LocalMinimizerIsGlobal}
	An initial state that locally minimizes a type-II quantity also globally minimizes it.
\end{thm}

This is because
type-II quantities are convex in the initial state.

For general processes, the nonlinearity of type-II objectives makes it difficult to find a closed-form expression for ideal inputs $\argmin_{\rho_0} f_{\rho_0}^{(\genop)} $.
Nevertheless, since type-II quantities are convex in the initial state, 
any number of simple algorithms, including gradient descent and related variations, are guaranteed to converge to the ideal input upon iteration.
By convexity, a local minimum in $f_{\rho_0}^{(\genop)} $, when minimizing over the set of density matrices,
is also the global minimum.

However, constrained optimization---restricting to the set of density matrices in this case---is non-trivial.
Although some techniques for gradient descent on the manifold of density matrices have been developed---see, e.g., Refs.~\cite{Gonca16, Banc20convex} and references therein---we
found a more direct solution
to our problem, which we provide in this section.  
Our resulting algorithm for gradient descent 
provides a 
quantum generalization of the Frank--Wolfe algorithm~\cite{Jagg13}, with
a simple 
analytically solvable direction for descent at each step.

Using techniques introduced in Ref.~\cite{Riec21_Initial}, 
we can analytically calculate the gradient of type-II expectation values
around any initial state, given any parametrization of the state space.
If we consider arbitrary infinitesimal changes in the initial generalized Bloch vector, then the partial derivative of $f_{\rho_0}^{(\genop)}$, with respect to each Bloch-vector component, can be expressed as
\begin{align}
	\frac{\partial}{\partial {b_0}_m} f_{\rho_0}^{(\genop)} = \tr(\Gamma_m \genop) + \tr(\Gamma_m \ln \rho_0)  - \tr \Bigl\{ \tr_\text{env} \bigl[ U ( \Gamma_m \otimes \rho_0^\text{env}) U^\dagger \bigr] \ln \rho_\tau \Bigr\} ~.
\end{align}
This follows via an adaptation of the derivation that led to Eqs.~(C16) and (O4)
in Ref.~\cite{Riec21_Initial}.

It is useful to notice how the elements of an arbitrary matrix basis can be expressed as a linear combination of the $d^2$ test inputs:
\begin{align}
	\Gamma_m = \sum_{n=1}^{d^2} (B^{-1})_{m+1, n} \, \rho_0^{(n)} ~,
\end{align}
where we have included the reference state as $\Gamma_0 = \sysref$.
We can likewise express 
the time-evolution of these matrices, 
as the same linear combination of the time-evolved test inputs:
\begin{align}
	\Gamma_m' \coloneqq
	\tr_\text{env} \bigl[ U ( \Gamma_m \otimes \rho_0^\text{env}) U^\dagger \bigr] = \sum_{n=1}^{d^2} (B^{-1})_{m+1, n} \, \rho_\tau^{(n)} ~.
	\label{eq:Gprime}
\end{align}
Partial derivatives of  $f_{\rho_0}^{(\genop)}$
can thus be calculated via 
\begin{align}
	\frac{\partial}{\partial {b_0}_m} f_{\rho_0}^{(\genop)} = \tr(\Gamma_m \genop) + \tr(\Gamma_m \ln \rho_0)  - 
	\tr(\Gamma_m' \ln \rho_\tau)
~,
	\label{eq:SimplePartialCalc}
\end{align}
where $\rho_\tau = \Gamma_0' +  \tfrac{1}{\constname} \sum_{n=1}^{d^2-1} \tr(\rho_0 \Gamma_n) \, \Gamma_n' $.

A gradient can be constructed as
the Hermitian operator
\begin{align}
	\vec{\nabla} f_{\rho_0}^{(\genop)} 
	= \sum_{m=1}^{d^2-1} \Gamma_m \,
	  \frac{\partial}{\partial {b_0}_m} f_{\rho_0}^{(\genop)} ~.
\end{align}
However, extra care must be taken to stay along the manifold of valid density matrices.
To achieve this, one can use the directional derivative
\begin{align}
\frac{\rho_0' - \rho_0}{ \| \rho_0' - \rho_0 \|} \cdot \vec{\nabla} f_{\rho_0}^{(\genop)} =
	\tr \bigl[ (\rho_0' - \rho_0)  \vec{\nabla} f_{\rho_0}^{(\genop)} \bigr] /  \| \rho_0' - \rho_0 \| 
	~,
\end{align}	
which is the linear change in $f_{\rho_0}^{(\genop)}$ at $\rho_0$
when moving in the direction of $\rho_0'$.
Since density matrices are a convex set,
a change in this direction is guaranteed to move along the manifold of density matrices.
For simplicity, $\| \cdot \|$ can be chosen to be the trace norm.

Recall that, since $f_{\rho_0}^{(\genop)}$ is convex over initial states,
descent 
always benefits its global minimization;
there are no non-global local minima to get stuck in.
If $\rho_0$ does not minimize $f_{\rho_0}^{(\genop)}$,
then we can find directions (along the manifold of density matrices)
with negative slope, and these directions will lead towards the minimizer.
Note that moving infinitesimally from $\rho_0$ towards $\rho_0'$ will reduce $f_{\rho_0}^{(\genop)}$
whenever 
$	\tr \bigl[ (\rho_0' - \rho_0)  \vec{\nabla} f_{\rho_0}^{(\genop)} \bigr] < 0$.
Accordingly, 
we can always find a valid direction of descent by identifying the 
$\rho_0'$ that minimizes
$	\tr \bigl[ (\rho_0' - \rho_0)  \vec{\nabla} f_{\rho_0}^{(\genop)} \bigr] $.
Fortunately, this desired $\rho_0'$---call it $\sigma_\text{min}$---can be found explicitly and analytically since
\begin{align}
\sigma_{\min}  & := 	
\argmin_{\rho_0'}
	\tr \bigl[ (\rho_0' - \rho_0)  \vec{\nabla} f_{\rho_0}^{(\genop)} \bigr]  \\
&=
\argmin_{\rho_0'}
\tr \bigl( \rho_0'  \vec{\nabla} f_{\rho_0}^{(\genop)} \bigr) \\
&= 
\frac{\ket{\xi} \bra{\xi}}{ \braket{\xi | \xi} }
~,
\end{align}	
where 
\begin{align}
\ket{\xi} = \argmin_{\ket{\lambda}} \braket{ \lambda | \vec{\nabla} f_{\rho_0}^{(\genop)}  | \lambda}
\end{align}
is the minimal-eigenvalue eigenstate of the Hermitian operator 
$ \vec{\nabla} f_{\rho_0}^{(\genop)} $.

Putting this all together, we can now propose a simple descent 
method for finding the optimal quantum state:

\begin{shaded}

\vspace{0.3em}

\Ovalbox{ \qquad \qquad \qquad \qquad \qquad \qquad \qquad \qquad \textbf{Algorithm to obtain} 
$\argmin_{\rho_0} f_{\rho_0}^{(\genop)}  \qquad \qquad \qquad \qquad \qquad \qquad \qquad \qquad \quad$}

\vspace{0.4em}

From $d^2$ linearly independent test inputs, 
record the initial Bloch matrix $B$, 
construct the thermodynamic operator $\mathcal{X}$,
and record the time-evolved test states $(\rho_\tau^{(n)})_{n=1}^{d^2}$
to obtain $(\Gamma'_{n})_{n=0}^{d^2-1}$.

Choose an arbitrary initial density matrix $\rho_{0}$, 
which will be updated iteratively.

\vspace{0.5em}

At each iterative step $k$:
\begin{enumerate}
\item
Calculate the gradient 
$ \vec{\nabla} f_{\rho_0}^{(\genop)} $  
via Eq.
\eqref{eq:SimplePartialCalc}.
\item 
Determine 
the descent direction $\hat{n} = \frac{ \sigma_{\min}  - \rho_0}{ \| \sigma_{\min}  - \rho_0\|}$, 
with $\sigma_{\min}  = \frac{\ket{\xi} \bra{\xi} }{ \braket{\xi | \xi} } $,
from the minimal-eigenvalue eigenstate $\ket{\xi}$ of the Hermitian operator 
$ \vec{\nabla} f_{\rho_0}^{(\genop)} $.
\item 
Update the initial state $\rho_0$ to 
approach $\argmin_{\rho_0} f_{\rho_0}^{(\genop)} $,
according to:
\begin{align}
	\rho_0 \mapsto \rho_0 - a_k \hat{n} \,
	\tr \bigl( \hat{n} \, \vec{\nabla} f_{\rho_0}^{(\genop)} \bigr)
	~,
\end{align}
where $a_k$ is a small positive value that 
diminishes
with large $k$.
\end{enumerate}
\end{shaded}

Notice that $\hat{n}$ is a traceless operator.
We take $a_k = \tfrac{2}{k+2}$, as suggested by the Frank--Wolfe algorithm~\cite{Jagg13}.
In the limit of many iterations, 
the algorithm converges then towards 
its unique fixed point $\alpha_0
= \argmin_{\rho_0} f_{\rho_0}^{(\genop)} $
with error on the order of 
$\mathcal{O}(1/k)$.

The minimizing input is generically a mixed state with full support.
In these cases, 
success in the type-II minimization can be verified through the 
mismatch theorem:
\begin{align}
	 f_{\rho_0}^{(\genop)} - f_{\alpha_0}^{(\genop)} = 
	 \text{D}[\rho_0 \| \alpha_0] - \text{D}[\rho_\tau \| \alpha_\tau] 
	 ~,
\end{align}
where 
$\text{D}[\rho \| \alpha] = \tr(\rho \ln \rho) - \tr(\rho \ln \alpha)$ is the quantum relative entropy, 
and $\alpha_0 = \argmin_{\rho_0} f_{\rho_0}^{(\genop)} $
is the ideal input~\cite{Riec21_Initial, Kolc21_State}.

\subsubsection{Inputs that maximize type-II functions}

\begin{thm}	
	\label{thm:Type2MaxIsPure}
	A pure state input maximizes a type-II quantity.
\end{thm}

This is because
type-II quantities are convex functions of the initial state, 
and the maximum of a convex function over a convex set (set of mixed states)
is achieved on the boundary of the convex set (the pure states).

Local maximization can be achieved by a slight adaptation of the above algorithm, by ascending rather than descending the gradient.  The only difference is that the ascent direction is obtained via the \emph{maximal}-eigenvalue eigenstate of the Hermitian operator $ \vec{\nabla} f_{\rho_0}^{(\genop)}$, whereas gradient descent used its minimal-eigenvalue eigenstate.

However, local maxima of type-II quantities are not necessarily global maxima.  In our simulations, we found the global maxima by seeding our algorithm with $\argmax_{\rho_0} \tr(\rho_0 \mathcal{X})$ 
which was found via spectral decomposition of the thermodynamic operator $\mathcal{X}$.

For completeness and clarity, 
we provide this ascent algorithm explicitly:

\begin{shaded}
	
	\vspace{0.3em}
	
\Ovalbox{ \qquad \qquad \qquad \qquad \qquad \qquad \qquad \qquad \textbf{Algorithm to obtain} 
$\argmax_{\rho_0} f_{\rho_0}^{(\genop)}  \qquad \qquad \qquad \qquad \qquad \qquad \qquad \qquad \quad$}
	
	\vspace{0.4em}

From $d^2$ linearly independent test inputs, 
record the initial Bloch matrix $B$, 
construct the thermodynamic operator $\mathcal{X}$,
and record the time-evolved test states $(\rho_\tau^{(n)})_{n=1}^{d^2}$
to obtain $(\Gamma'_{n})_{n=0}^{d^2-1}$.

Choose
the initial density matrix to be  $\rho_{0} = \argmax_{\rho_0} \tr(\rho_0 \mathcal{X})$, 
which will be updated iteratively.

\vspace{0.5em}

At each iterative step $k$:
\begin{enumerate}
	\item
	Calculate the gradient 
	$ \vec{\nabla} f_{\rho_0}^{(\genop)} $  
via Eq.~\eqref{eq:SimplePartialCalc}.
	\item 
	Determine 
	the ascent direction $\hat{n} = \frac{ \sigma_{\max} - \rho_0}{ \| \sigma_{\max}  - \rho_0\|}$, 
	with $\sigma_{\max}  = \frac{\ket{\psi} \bra{\psi} }{ \braket{\psi | \psi} } $,
from the maximal-eigenvalue eigenstate $\ket{\psi}$ of the Hermitian operator 
	$ \vec{\nabla} f_{\rho_0}^{(\genop)} $.
	\item 
	Update the initial state $\rho_0$ to 
	approach $\argmax_{\rho_0} f_{\rho_0}^{(\genop)} $,
	according to:
	\begin{align}
		\rho_0 \mapsto \rho_0 + a_k \hat{n} \,
		\tr \bigl( \hat{n} \, \vec{\nabla} f_{\rho_0}^{(\genop)} \bigr)
		~,
	\end{align}
	where $a_k$ is a small positive value that diminishes
	with large $k$.
\end{enumerate}
In the limit of many iterations, 
the algorithm converges towards 
a local maximum of
$f_{\rho_0}^{(\genop)} $.

\end{shaded}

\subsection{Overwriting processes} 
\emph{Overwriting processes}---processes that overwrite the physical state of the system---are an important class of processes including 
memory reset,
state preparation, 
work extraction,
and processes that lead to either equilibrium or nonequilibrium steady states.
We find that their thermodynamically ideal inputs can all be found directly and analytically.

For reliable overwriting processes, 
for which the final state $r_\tau$ is very nearly independent of the input,
$S(\rho_\tau) = S(r_\tau)$ will effectively be a constant.
In this case,
as shown in   App.~\ref{sec:OperatorExpressionForType2},
$f_{\rho_0}^{(\genop)}$ can be expressed as
$f_{\rho_0}^{(\genop)} = 
\text{D}[\rho_0 \| \omega^{(\genop)}] - \ln[ \tr(e^{-\genop}) ] + S(r_\tau)$ 
where $\omega^{(\genop)} \coloneqq e^{-\genop} / \tr(e^{-\genop})$.
Thus, for reliable overwriting processes,
it is clear that 
$\omega^{(\genop)}$
uniquely minimizes $f_{\rho_0}^{(\genop)}$.
It is worth noting that, 
if $\genop$ is a bounded finite-dimensional operator, 
$\omega^{(\genop)} $ has full rank.
Physically, this tells us that 
$\omega^{(\genop)} $
is a non-pure mixed state.

\begin{thm}
\label{thm:Type2MinimizerForReset}
	For any overwriting process, 
	the expectation value of any type-II quantity 
	$\braket{X}_{\rho_0} = p f_{\rho_0}^{(\genop)}
	= p[\tr(\rho_0 \genop) + S(\rho_\tau) - S(\rho_0) ]$ 
	is uniquely minimized by
	the mixed-state input
\begin{align}
\omega^{(\genop)} = e^{-\genop} / \tr(e^{-\genop}) ~.
\end{align}
The corresponding minimal value is
$\braket{X}_{\omega^{(\genop)}} = p S(r_\tau) - p  \ln[ \tr (e^{-\genop}) ] $,
where $r_\tau$ is the input-independent final state of the overwriting process.
\end{thm}

Recall that the thermodynamic operator can be expressed 
as 
$\genop = 
\braket{X}_{\sysref} I +  \vec{x} \cdot \vec{\Gamma} / \constname$.
In terms of the thermodynamic vector $\vec{x}$, we find that 
\begin{align}
	\omega^{(\genop)} = e^{-  \vec{x} \cdot \vec{\Gamma} / \constname } / \tr(e^{- \vec{x} \cdot \vec{\Gamma} / \constname }) ~.
\end{align}

For a qubit in the standard Pauli-matrix basis $\vec{\Gamma} = \vec{\sigma}/2$,
this further reduces to
\begin{align}
	\omega^{(\genop)}
&= e^{-\vec{x} \cdot \vec{\sigma}} / \tr(e^{-\vec{x} \cdot \vec{\sigma}}) 
	\\
	&
	= 
	I/2 - \frac{\hat{x} \cdot \vec{\sigma}}{2} \tanh x
	~.
	\label{eq:MinDissQubit}
\end{align}
Intuitively, Eq.~\eqref{eq:MinDissQubit} tells us that the 
minimizing
Bloch vector $\vec{a}^* = - \tanh(x) \hat{x}$ points in the opposite direction of
the thermodynamic vector $\vec{x}$ to reduce entropy flow or promote energy gain,
as the case may be. 
However, this tendency to reduce heat or increase energy 
is balanced against the entropy gain incurred when
tarnishing a pure state. 
As the magnitude of the thermodynamic vector $x$ grows
beyond unity, the minimizing Bloch vector converges
exponentially to the edge of the Bloch sphere.

For example,
for any reliable overwriting process operating on a qubit, 
this tells us that
the unique initial state leading to minimal entropy production is 
$I/2 - \frac{\hat{\varphi} \cdot \vec{\sigma}}{2} \tanh \varphi$, in agreement~\footnote{For comparison with the previous work, note that the entropy flow vector $\vec{\varphi}$ as defined here is scaled by a factor of 1/2 relative to the $\vec{\phi}$ of Ref.~\cite{Riec21_Impossibility}: $\vec{\varphi} = \vec{\phi}/2$.  This rescaling has a number of aesthetic benefits, and leads more naturally to our high-dimensional generalization.} 
with 
Ref.~\cite[Eq.~(25)]{Riec21_Impossibility} obtained by different means.
Whereas 
Ref.~\cite{Riec21_Impossibility} identified the initial state of a qubit 
that would lead to minimal entropy production
during any qubit-reset process, the current result
provides a significant generalization.
We now identify
the ideal input to any overwriting process in any finite dimension for any type-II objective,
including maximizing free-energy gain.

While the initial states leading to 
minimal entropy production or maximal gain in free energy
are non-trivial mixed states,
the input leading to the largest reduction in entropy is always the same for any overwriting process.
Notice that the largest reduction in entropy is achieved by the fully mixed input state
$\omega^{(0)} = I/d$
for any overwriting process in any dimension $d$.

We have so far found the ideal inputs to \emph{minimize} the expectation values of type-II quantities during an overwriting process---for example, the inputs leading to minimal entropy production or maximal increase in nonequilibrium free energy.  
It is just as natural and important to ask: 
Which inputs \emph{maximize} the expectation values of type-II quantities, leading, e.g., to maximal entropy production or the biggest reduction in nonequilibrium free energy?

Since $S(\rho_\tau) = S(r_\tau)  $ is independent of the input to an overwriting process, maximizing  
$f_{\rho_0}^{(\genop)}
= \tr(\rho_0 \genop) + S(\rho_\tau) - S(\rho_0) $
asks us to simultaneously minimize $S(\rho_0)$
and maximize  $\tr(\rho_0 \genop) $.
Notice that $S(\rho_0)$ is minimized for any pure state,
whereas $\tr(\rho_0 \genop) $ is maximized by some pure state, according to Thm.~\ref{thm:Type1ExtremaArePure}.
Accordingly, the two objectives can be simultaneously satisfied
by identifying the pure state that maximizes $\tr(\rho_0 \genop) $.
We thus inherit an answer from Thm.~\ref{thm:Type1ExtremaArePure}, 
when seeking to maximize the expectation value of a type-II quantity.

\begin{thm}
\label{thm:ForReset_Type2MaxIsType1Max}	
	For any overwriting process, 
	the expectation value of any type-II quantity 
	$\braket{X}_{\rho_0} = p f_{\rho_0}^{(\genop)}
	= p[\tr(\rho_0 \genop) + S(\rho_\tau) - S(\rho_0) ]$ 
	is maximized by
	a pure-state input 
	that maximizes $\tr(\rho_0 \genop )$.
	This is satisfied by any eigenstate of $\genop$ with maximal eigenvalue.
\end{thm}

For example, for any overwriting process, 
the input that maximizes entropy production 
is the pure state that maximizes entropy flow, 
which is an eigenstate of the expected-entropy-flow operator 
$\EFO$ with maximal eigenvalue.
Similarly, the biggest reduction in nonequilibrium free energy will be achieved for any overwriting process from the pure-state
input that maximizes the reduction in energy,
which is an eigenstate of the change-of-energy operator $\mathcal{Q}+ \mathcal{W}$ with minimal eigenvalue.

\subsection{Perturbative correction to type-II minimizers}

Here we 
find a closed-form expression for the approximate minimizer of any type-II functional---e.g., the input state that minimizes entropy production or maximizes the change in nonequilibrium free energy---for any process,
via a second-order expansion of the type-II functional around
a reference state.  
This method becomes exact in the limit that
the true minimizer is a small perturbation from the reference state---e.g., in the case of small changes to a protocol with known minimizer.

To obtain our result, we 
minimize $f_{\rho_0}^{(\genop)} $
over
the space of 
valid
generalized Bloch vectors.
It is convenient to introduce the function $\tilde{\rho}$
with $\pbb =  I/d + \bb \cdot \vec{\Gamma} $
which maps Bloch vectors to their corresponding density matrices,
and the function $\tilde{\rho}'$ 
with $\pbbF =  \Gamma_0' + \vec{b} \cdot \vec{\Gamma}' $
which maps initial Bloch vectors (at time 0) to their corresponding time-evolved density matrices at time $\tau$.
To simplify notation,
we introduce the function $\tilde{f}$ such that $\tilde{f}(\vec{b}) =  f_{\pbb}^{(\genop)}$.

We seek
the optimal $\bb^{*} = \argmin_{\vec{b}} \, 
\tilde{f}(\bb)$
in the \emph{perturbative regime}, where
the optimal $\pbb[\bb^{*}]$ is a small perturbation
away from some reference initial state $\pbb[\pp]$. For instance,
$\tilde{f}$ may represent entropy production and $\pbb[\pp]$ may be an initial equilibrium state
such that $\pbb[\bb^{*}]\approx\pbb[\pp]$ in the limit of slow
time-dependent driving. 

Assume that the density matrix $\rho(\pp)$ is positive definite.
By continuity, for any $\vec{b}$
sufficiently close to $\pp$, $\pbb$ is also positive definite and therefore
a valid density matrix.~\footnote{Hermiticity and trace-one conditions are automatically satisfied by any generalized Bloch vector, so we only need to be careful about ensuring positive semi-definiteness of the induced matrix.} 
Now expand $\tilde{f}$ to second order in $\vec{\epsilon} \coloneqq  \bb-\pp$,
\begin{equation}
	\tilde{f}(\bb)\approx \tilde{f}(\pp) + \vec{\epsilon}^{\top}\vec{j}+\frac{1}{2}\vec{\epsilon}^{\top} \Hess \vec{\epsilon},\label{eq:taylor}
\end{equation}
where $\vec{j}$ is the gradient vector and $\Hess$ is the Hessian matrix of $\tilde{f}$ evaluated
at $\pp$. The elements of the gradient vector are given by
\begin{align}
	j_{n} & = \partial_{b_{n}} \tilde{f} \, \bigr|_{\pp} \nonumber \\
	& =\tr ( \Gamma_{n} \mathcal{X} ) + \tr \bigl[ \Gamma_{n} \ln \pbb[\pp] \bigl] - \tr \bigl[ \Gamma_{n}' \ln\pbbF[\pp] \bigr] ~. \label{eq:grad1}
\end{align}
The elements of the Hessian matrix
are $\Hess_{m,n} = (\partial_{b_{m}} \partial_{b_{n}} f ) |_{\vec{\pi}}$.   The Hessian follows by considering
the second derivatives of the von Neumann entropy,
which we calculate in App.~\ref{sec:SecondDerivativeOfEntropy}
to find 
\begin{equation}
\partial_{b_{m}}\partial_{b_{n}} S (\tilde{\rho}) 
\bigr|_{\pp}=
	- \sum_{k,\ell}\phi(\nu_{k},\nu_{\ell})\langle k \vert\Gamma_{n}\vert \ell \rangle\langle  \ell \vert\Gamma_{m}\vert k  \rangle
	~,
	\label{eq:d2}
\end{equation}
where we used the eigendecomposition $\pbb[\pp]=\sum_{k} \nu_{k} \vert k \rangle\langle k \vert$
and defined $\phi(a,b)\coloneqq \frac{\ln a-\ln b}{a-b}$ (with $\phi(a,a)=1 / a$
by continuity), which is the reciprocal of the logarithmic mean. 
We note 
that this second-order expansion
has some resemblance to
the Kubo--Mori--Bogoliubov metric
of quantum-information geometry~\cite{Petz93_Bogoliubov}. 
We derive
a similar expression for
$\partial_{b_{m}}\partial_{b_{n}}S(\tilde{\rho}')  $, which
gives the following form
for the matrix elements 
of the Hessian:
\begin{align}
	\Hess_{m,n} = & 
	\biggl[
	\sum_{k, \ell } \phi(\nu_{k},\nu_{\ell})\langle k \vert \Gamma_{n} \vert \ell \rangle\langle \ell \vert\Gamma_{m}\vert k \rangle 
	\biggr]
-
	\biggl[
	\sum_{k, \ell } \phi(\nu_{k}',\nu_{\ell}')\langle k'\vert\Gamma_{n}'\vert \ell'\rangle\langle \ell'\vert\Gamma_{m}'\vert k'\rangle \biggr] ~,
	\label{eq:hess1}
\end{align}
where we used the eigendecomposition $\pbbF[\pp]=\sum_{k}\nu_{k}'\vert k'\rangle\langle k'\vert$.
Note that $\Hess$ is positive semi-definite since $\tilde{f}$ is a convex function.
For simplicity, assume for now that $\tilde{f}$ is strictly convex,
in which case $\Hess$ is positive definite and has an inverse $\Hess^{-1}$.
(We give the generalization later.)

Finally, we minimize Eq.~\eqref{eq:taylor} in closed form. First, complete
the square to write
\begin{align}
	\vec{\epsilon}^\top \vec{j}+\frac{1}{2}\vec{\epsilon}^\top \Hess \vec{\epsilon} & =\frac{1}{2}(\vec{\epsilon}+ \Hess^{-1}\vec{j})^\top \Hess (\vec{\epsilon}+ \Hess^{-1}\vec{j})-\frac{1}{2}\vec{j}^\top \Hess^{-1}\vec{j}\\
	& \ge-\frac{1}{2}\vec{j}^\top \Hess^{-1}\vec{j} ~,
\end{align}
where the last inequality is achieved by setting $\vec{\epsilon}^{*}=-\Hess^{-1}\vec{j}$.
This implies that in the perturbative regime, the optimal state
is given by
\begin{align}
\pbb[\bb^{*}]=\pbb[\pp]-(\Hess^{-1}\vec{j})\cdot \vec{\Gamma} ~,
\end{align}
which achieves the optimal value 
\begin{align}
\min \tilde{f}(\bb)= \tilde{f}(\pp)-\frac{1}{2}\vec{j}^\top \Hess^{-1}\vec{j} ~.
\end{align}

The above expressions for the optimal state generalize
to allow for a singular Hessian
if we replace $\Hess^{-1}$
with  the Drazin inverse $\Hess^{\mathcal{D}}$ (or the group inverse, since we can diagonalize the Hessian)~\cite{Riec18_Beyond}.

\section{Generalized Bloch vectors, thermodynamic operators, and ideal inputs within a restricted subspace}
\label{sec:Ideal_restriction}

It will often be useful to know the ideal input to a device,
within a restricted subspace of possible inputs.
For example, when operating on a system with a countably infinite number of energy eigenstates, we may care about inputs with non-zero probability amplitude only in the lowest $N$ energy eigenstates, for some finite $N$.
Or, we may be interested in the best classical inputs to a device, when coherent states cannot be readily prepared.
In such cases, we can find the thermodynamically ideal input to the transformation within this finite-dimensional subspace 
via a very straightforward adaptation of the above techniques.

Let $\mathcal{P}_P$ be the convex subspace of 
the system's possible density matrices $\mathcal{P}$, induced by the 
set of orthogonal projectors $P = \{ \Pi_j \}_j$ with $\Pi_j \Pi_k = \delta_{j,k} \Pi_j$, such that
\begin{align}
\mathcal{P}_P \coloneqq \Bigl\{ \rho \in \mathcal{P} : \rho = \sum_{\Pi \in P} \Pi \rho \Pi \Bigr\} ~.
\end{align}
Density matrices in this subspace 
act on a $d_P$-dimensional vector space $\mathcal{V}_P$,
where $d_P = \sum_{\Pi \in P} \tr(\Pi)$.
The identity operator on $\mathcal{V}_P$
is given by $I_P = \sum_{\Pi \in P} \Pi$.

Each projector $\Pi_j \in P$
has an associated dimension $d_j = \tr(\Pi_j)$.
The restricted subspace $\mathcal{P}_P$ is spanned by a basis of $L$ linearly independent density matrices, where $L = \sum_{j=1}^{|P|} d_j^2$.
In general, 
$d_P \leq L \leq d_P^2$.

By creating an appropriate operator basis and generalized Bloch vector for initial density matrices restricted to $\mathcal{P}_P$, we can easily adapt and leverage all results of previous sections
in this restricted setting.
To achieve this, the following development closely parallels the previous introduction of generalized Bloch vectors.

When the initial density matrix $\rho_0$
is restricted to $\mathcal{P}_P$,
it is useful to 
choose a complete basis 
$(I_P / d_P, \Gamma_1, \Gamma_2, \dots \Gamma_{L - 1})$ for linear operators acting on $\mathcal{V}_P$,
such that the Hermitian operators $\Gamma_n$ are all traceless and mutually orthogonal, satisfying
$\tr( \Gamma_n) = 0$
and 
$\tr(\Gamma_m \Gamma_n) =  \constname_P \, \delta_{m,n}$,
where we will choose the normalizing constant to be $\constname_P = \tfrac{d_P-1}{d_P}$.
Any initial density matrix then has a unique decomposition in the 
operator basis 
$\vec{\Gamma}_P = (\Gamma_1, \Gamma_2, \dots \Gamma_{L - 1})$, 
described by the 
generalized Bloch vector $\vec{b}_0 \in \mathbb{R}^{L - 1}$
via 
\begin{align}
	\rho_0 = I_P / d_P  +  \vec{b}_0 \cdot \vec{\Gamma}_P ~.
	\label{eq:RestrictedBlochExpansion}
\end{align}
Since the magnitude of the
Bloch vector is 
$b_0 = \sqrt{\frac{\tr(\rho_0^2) d_P - 1}{d_P-1}}$,
the density matrix represents a pure state iff
the magnitude of its corresponding Bloch vector is one.

For concreteness, 
we can 
choose the ordered operator basis to begin with a scaled ordering of the $d_P$
diagonal generalized Gell-Mann matrices,
followed by 
$d_j (d_j - 1) / 2$ non-diagonal symmetric Gell-Mann matrices
and  
$d_j (d_j - 1) / 2$ antisymmetric Gell-Mann matrices
for each projector $\Pi_j \in P$
with $d_j > 1$.
Recall that these matrices 
are given explicitly in
App.~\ref{sec:GGM}.

It is now productive to consider the restriction of a thermodynamic operator to this subspace 
$\genop_P \coloneqq \sum_{\Pi \in P} \Pi \genop \Pi$.
It is easy to check that the expectation value for any type-I thermodynamic quantity satisfies
\begin{align}
\tr(\rho_0 \genop_P) = \tr(\rho_0 \genop) = \braket{X}_{\rho_0} \quad \text{for all } \rho_0 \in \mathcal{P}_P ~.	
\end{align}
Notably, $\genop_P$ can be constructed directly, just as $\genop$ was in the previous sections---but
now in the finite-dimensional restricted subspace $\mathcal{P}_P$, using the operator basis $\vec{\Gamma}_P$,
generalized Bloch vector, 
and thermodynamic operator
$\vec{x}_P$
associated with this restricted subspace.
For general processes,
the thermodynamic operator restricted to this subspace 
is constructed as
\begin{align}
\genop_P = \braket{X}_{I_P / d_P}  I_P +  \vec{x}_P \cdot \vec{\Gamma}_P / \constname_P ~.
\end{align}

The thermodynamically ideal inputs within the subspace of interest
are constructed just as before, with
the extremal eigenvalues and associated eigenstates of $\genop_P$ playing the special role indicated in Theorems \ref{thm:Type1ExtremaArePure} and \ref{thm:ForReset_Type2MaxIsType1Max}.
For overwriting processes, the minimal value of type-II quantities is achieved within subspace $\mathcal{P}_P$ by
$\omega^{(\genop_P)}
= e^{-\genop_P} / \tr(e^{-\genop_P})
= e^{- \vec{x}_P \cdot \vec{\Gamma}_P / \constname_P } / \tr(e^{- \vec{x}_P \cdot \vec{\Gamma}_P / \constname_P }) $,
which 
extends an analogous result of 
Ref.~\cite{Kolc21_State} that was formulated for the case of entropy production.

It is interesting to note that these methods---restricting to the $d_P$-dimensional subspace---work even though the states are not restricted to this subspace during their evolution.  Indeed, the protocol can spread these initially-restricted inputs across infinite dimensions, but our finite-dimensional inference of the ideal input within this initial subspace remains valid.

To pursue the example of inputs with support restricted to the $N$ lowest-energy eigenstates $\{ \ket{E_n } \}_{n=1}^N$ of some
initial system Hamiltonian $H_0$, we would consider the density matrices with support on the system's Hilbert subspace:
$\mathcal{H}_\text{sub} = \text{span} \bigl( \{ \ket{E_n } \}_{n=1}^N \bigr)$.
I.e., we would consider the restricted set of initial density matrices
 $\mathcal{P}_\text{sub} = \Bigl\{ \rho = \sum_\ell p_\ell \frac{\ket{\psi_\ell} \bra{\psi_\ell}}{ \braket{ \psi_\ell | \psi_\ell }}
 : \, p_\ell \in(0, 1] , \, \sum_\ell p_\ell = 1 , \, \ket{\psi_\ell} \in \mathcal{H}_\text{sub}  \Bigr\}$,
which is induced by the single rank-$N$ projector $\Pi = \sum_{n=1}^N \ket{E_n} \bra{ E_n} $.
Within this subspace, there are 
$L = d_{\{ \Pi \}}^2 = N^2$ linearly independent
initial states (compared to the $d^2 = \infty$ linearly independent initial states within the full Hilbert space).

\subsection{Classical systems}

If we define classical inputs as those states restricted to be initially incoherent in a particular `classical' orthonormal basis $\mathcal{C}$,
then we can see how the general quantum problem 
simplifies significantly,
if we seek the classical thermodynamic operators and ideal classical inputs.
Classical density matrices are those induced by the set of orthonormal
rank-1 projectors
$P = \{ \ket{b} \! \bra{b} \}_{b \in \mathcal{C}}$.
Classical density matrices are diagonal in the classical basis, and can be regarded as a representation of a probability distribution over these classical states. 
Note that for $d$-dimensional systems, there are only $L=d$ linearly independent classical density matrices, in contrast to the $d^2$ linearly independent quantum density matrices over the same space.
Accordingly, each classical thermodynamic vector has 
$d-1$ components, 
rather than the $d^2-1$ components of its quantum counterpart.
Similarly,
each classical thermodynamic operator will be fully determined by the behavior of $d$, rather than $d^2$, initial states.

\section{Example 1: Nonequilibrium thermodynamics of a  qubit-reset device}
\label{sec:Examples}

Quantum computing requires a mechanism for resetting each qubit to the computational-basis state $\ket{0} = \sigma_z \ket{0}$.
Different implementations
of the same task will however have distinct sets of thermodynamically ideal inputs.
Nonequilibrium thermodynamic
quantities are determined less by
\emph{what} you do 
than \emph{how} you do it.

For a paradigmatic illustration of 
our results,
we consider the same
device for qubit reset that was 
used 
in Ref.~\cite{Riec21_Impossibility}
to generate Ref.~\cite{Riec21_Impossibility}'s Fig.~1.  The device works by changing 
both the energy gap and spatial orientation of the energy eigenstates of the qubit, while the qubit is weakly coupled to a thermal environment.
The detailed dynamical equations are reviewed in App.~\ref{sec:EOMforReset}, but these details are not central to the main points we make here.

To determine 
how the device's thermodynamic behavior
depends on the input state,
we track the evolution of 
four randomly sampled initial density matrices, together with the thermodynamic output from each of these four inputs.
From the Bloch matrix $B$
and the measured thermodynamic output,
we construct the Hermitian thermodynamic operators.
For example, 
the expected-heat and expected-work operators,
$\mathcal{Q}$ and $\mathcal{W}$,
allow us to
determine 
(i) 
the ideal inputs leading to minimal and maximal heat and work,
and (ii)
the full range of heat and work that can be attained by \emph{any} input to the device.
These are obtained from
the extremal
eigenvalues and associated eigenstates of the thermodynamic operators.

\subsection{Diversity among ideal inputs for thermodynamic objectives}

Simple combination and manipulation of the heat and work operators reveals the 
diversity of
ideal inputs 
for a multitude of 
different thermodynamic objectives, as shown in Fig.~\ref{fig:IdealQubitInputs}.

In the figures, we emphasize the heat exhausted to the environment $-Q$, rather than the heat $Q$ directly.
In this example, with a single environmental bath at constant temperature $T$, note that entropy flow to the environment is simply related to heat out of the system via $\EF = -Q/T$.
The expected-entropy-flow operator 
is thus simply related to the expected-heat operator in this case, 
via $\EFO = - \mathcal{Q} / T$.
Meanwhile, the expected-energy-change operator
is simply 
$\mathcal{Q} + \mathcal{W} $.

\begin{figure}[H]

\begin{center}
	\includegraphics[width=\textwidth]{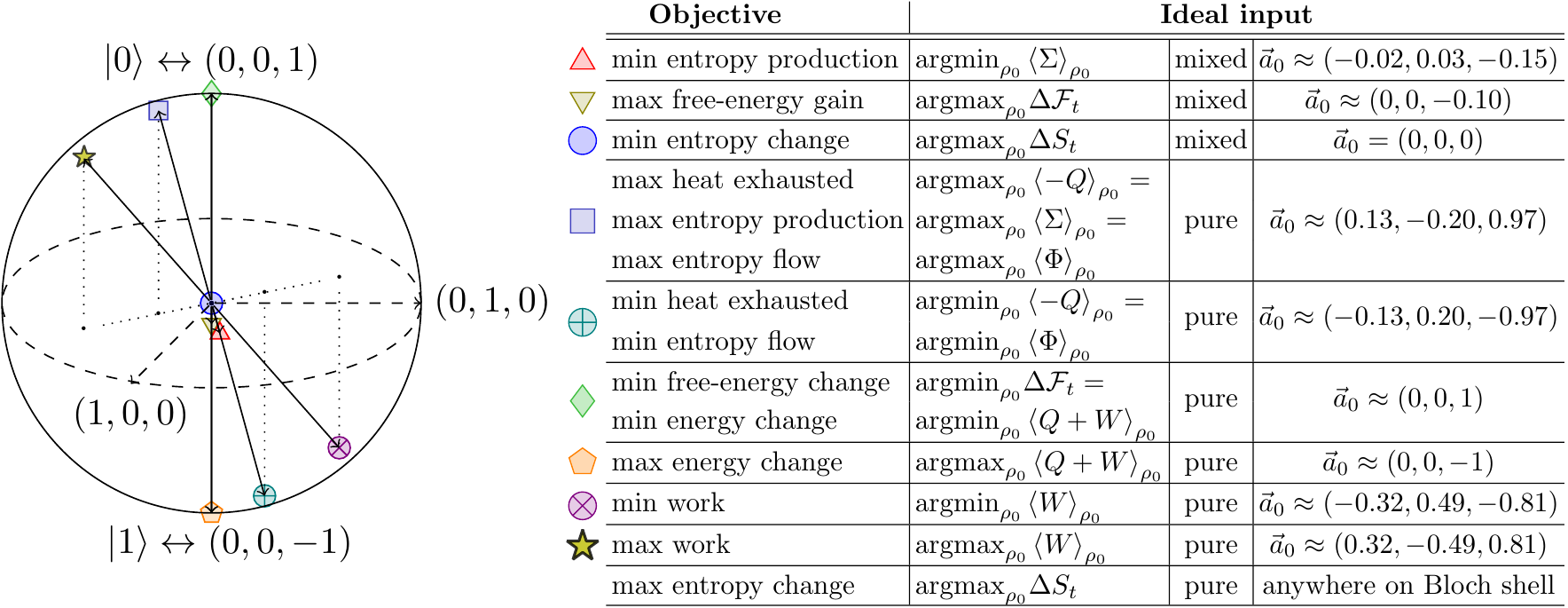}
\end{center}

\caption{Diversity of ideal inputs for a finite-time qubit-reset process, displayed on and in the Bloch sphere.  The states extremizing heat, work, and energy-change all lie on the surface of the Bloch sphere, in the direction of a maximal eigenstate of the corresponding thermodynamic operators.  The entire surface of the Bloch sphere maximizes entropy gain.  Minimal entropy production and maximal free energy gain are achieved by non-trivial mixed-state inputs. 
The change in entropy is minimized by the fully-mixed input.
Entropy production is maximized by the same pure-state input that 
maximizes heat exhaustion.
The greatest loss of free energy occurs for the same pure-state input that loses the most energy.
}
\label{fig:IdealQubitInputs}
\end{figure}

\subsection{Bounding the behavior of all inputs}

Continuing our example of the qubit-reset dynamics, we now leverage our results to identify the extremal thermodynamic behavior that can be attained by any input throughout the process.

\begin{figure}[h]

	\begin{center}
		\includegraphics[scale=1.0]{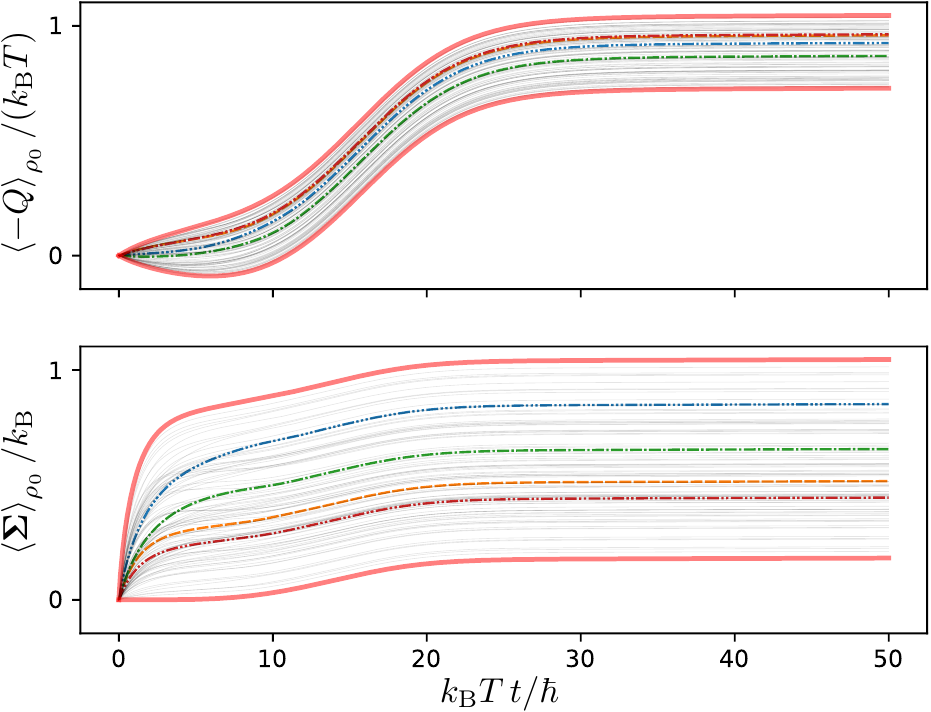}
	\end{center}
	
	\caption{Tracking the behavior of four inputs is enough to bound the behavior of all other inputs to a qubit process.  
		Here we show the the range of expectation values for exhausted heat and entropy production throughout a finite-time qubit-reset process.
		The expectation values from four random inputs are shown as dashed lines.  This allows construction of the thermodynamic operator $\mathcal{Q}$ throughout time.  
		(top) 
		Maximal and minimal heat, corresponding to extremal eigenvalues of $\mathcal{Q}$, shown as thick red solid lines;
		(bottom)
		Maximal and minimal entropy production, obtained from gradient descent/ascent, shown as thick red solid lines.
		These extrema
bound the behavior of all other inputs, 
		including the behavior of 100 other random initial conditions shown as thin gray solid lines.
}
	\label{fig:EdgeTracking}
\end{figure}

Fig.~\ref{fig:EdgeTracking}
demonstrates that thermodynamic observations from just four inputs yield the full range of thermodynamic behavior from any input.
For example, the min and max expected work at any time $t \in [0, \tau]$, obtainable from alternative inputs, is determined by the expected-work operator at that time.  The expected-work operator at any time is constructed from the expected work performed on each of the four test inputs up to that time.  Determining the ranges of work, energy change, and heat thus reduces to 
determining
eigenvalue ranges of the respective Hermitian operators.

Determining the range of entropy production throughout the process
is somewhat more complicated, although it still only requires the data from four test inputs.  
Notably,  
in the bottom panel of Fig.~\ref{fig:EdgeTracking},
we find
the states of minimal and maximal entropy production at times before the state is fully reset.
This employs the gradient-descent algorithm
developed in Sec.~\ref{sec:GradientDescentAlgo}.

\section{Example 2: Extracting work from a spatially extended state}
\label{sec:Example2}

\begin{figure}[h]
	\begin{center}
	\includegraphics[scale=0.935]{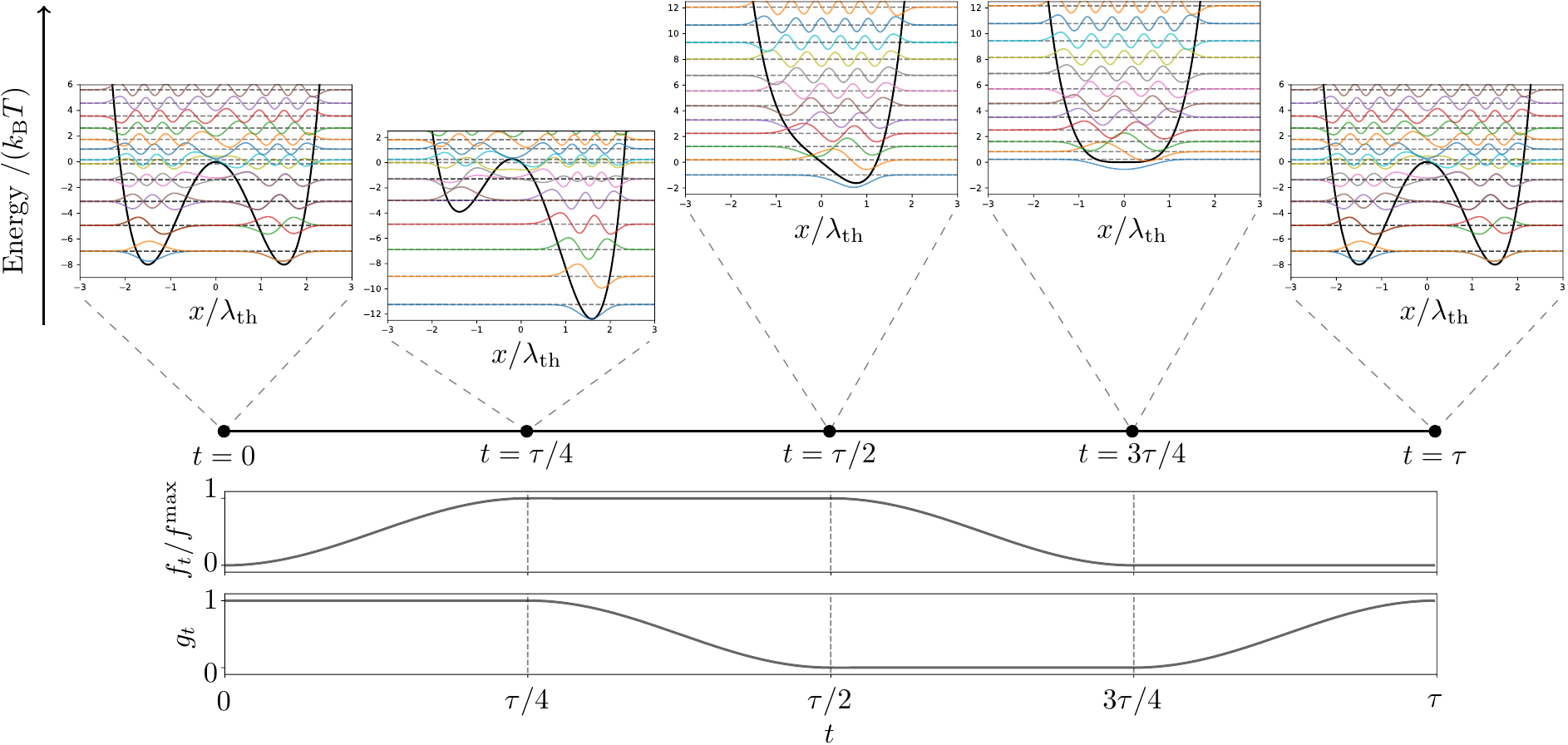}
	\end{center}
	\caption{
		Control protocol to extract work from the nonequilibrium state of a particle on the right-hand side of a double-well potential, via a time-varying one-dimensional potential-energy landscape in the quantum regime.  Middle and bottom panel show the protocol for tilting the potential and lowering the energetic barrier, respectively.
		The top panels show the potential-energy landscape at five representative times along the protocol, staggered so that they share the same energy axis.  We show a spatial representation of the lowest-lying energy eigenstates in these potentials, at the height of their respective energy eigenvalues.
}
\label{fig:double_well_protocol}
\end{figure}

Let us now turn to the analysis of an infinite-dimensional quantum system:
a charged particle in a time-dependent potential, in the presence of electromagnetic background radiation.
This can be interpreted
as an idealized model of a double quantum dot in the single-particle zero-current regime.
In this example, 
we use time-dependent Lindbladian dynamics, with a time-modulated double-well of potential energy across one spatial dimension.
The protocol is 
capable of extracting work from some states that are initially localized in the right well.
We consider the ideal input (for the task of work extraction),
and how it compares to the performance of other inputs,
within the Hilbert space spanned by the eight lowest lying 
energy eigenstates of the initial Hamiltonian.

We will quantify work extraction simply as the negative work performed.  In principle, this work can be extracted by a more explicit mechanism, like proposed in Ref.~\cite{Skrz14}.

Given any protocol for work extraction,
we can construct the expected-work operator 
to determine its breadth of behavior on all possible inputs.  The device will be best at extracting work from the 
lowest lying eigenstate of the expected-work operator; 
the corresponding minimal-eigenvalue magnitude
reports the maximal possible expected value of extractable work. 

In particular, we simulate a work-extraction protocol via the time-dependent system Hamiltonian
$H_t = \frac{p_x^2}{2m} + V_t$
where 
$p_x = -i\hbar \partial_x$ is the momentum operator and 
$m$ is the mass of the single-particle system.
The time-dependent potential-energy landscape is
\begin{align}
	V_t = 
16 h_0 \Bigl( \frac{x}{w_0} \Bigr)^4 
- 8  h_0 \Bigl( \frac{x}{w_0} \Bigr)^2 g_t 
- h_0 \frac{x}{w_0}  f_t
	~, 
\end{align}
where 
$\barrier_t$ and
$\tilt_t$
are non-negative scalar functions of time,
with $\barrier_0 = \barrier_\tau = 1$
and 
$\tilt_0 = \tilt_\tau = 0$.
To enter an interesting thermodynamic regime, we choose
the initial energy barrier between the two wells to be $h_0 = 8 \kB T$,
where $\kB T = 1/\beta$ is the thermal energy of the environment.
To enter an interesting quantum regime, we choose
the initial separation between the bottom of the two wells to be $w_0 = 3 \lambda_\text{th}$, where 
$\lambda_\text{th} = \hbar \sqrt{2 \pi \beta / m }$ is the thermal de Broglie wavelength.
To induce nonequilibrium quantum dynamics,
we choose a sufficiently fast protocol.

As depicted in Fig.~\ref{fig:double_well_protocol},
$\barrier_t$ controls the height of the barrier
while $\tilt_t$ controls the tilt of the energetic landscape throughout the protocol.
For a semiconductor-based double quantum dot, the barrier height and tilt could plausibly be modulated by applied voltages, gating the barrier and the bias across the device, respectively~\cite{DiVi05_Double}.
Note that the protocol is cyclic, since the initial and final Hamiltonian 
are both the same symmetric double-well potential.
During the protocol, the potential is tilted to the right, the barrier is lowered, the potential is untilted, and the barrier is then reintroduced, in that order.~\footnote{This work-extraction protocol is the time-reversed control sequence of a reset protocol.}  The exact form of the control protocol is given in App.~\ref{sec:DoubleWellDynamics}.

The system is coupled to a thermal environment throughout the protocol.  For our demonstration, 
we 
model
the excitation and relaxation dynamics 
when the system 
interacts with
the photon bath
of an ideal blackbody at temperature $T$.
The resulting dynamics
of the system during the protocol can be described by a Lindblad master equation 	
$\dot{\rho}_t = \mathcal{L}_{t} (\rho_t)$, 
given explicitly in App.~\ref{sec:DoubleWellDynamics}, which satisfies detailed balance 
in bath-mediated transitions between instantaneous energy eigenstates.

We simulate a regime that could plausibly be  
demonstrated in a room-temperature laboratory experiment
using a nanofabricated device.
The resulting dynamics exhibits a separation of timescales,
where the relaxation dynamics are much slower than the timescale of coherent oscillations between energy eigenstates.
(Dissipative transitions occur approximately only once per 100 million periods of phase oscillation between energy levels with a $\kB T$ energy spacing.)
We develop special methods 
to efficiently simulate the quantum dynamics (App.~\ref{sec:QuantumDynamicsOfTimescaleSep}) 
and thermodynamics 
(App.~\ref{sec:QuantumThermoDynamicsOfTimescaleSep}) 
over the extended duration $\tau = 4 \times 10^9 \beta \hbar \approx 100 \, \mu$s 
of the nonequilibrium protocol.

Which initial state yields the most extracted work given this protocol?
To determine this, 
we construct the expected-work operator $\mathcal{W}$ acting on the Hilbert space 
spanned by the 
eight lowest lying initial energy eigenstates,
via simulating the behavior of 64 random linearly independent initial states within this subspace. 
In our simulations,
these initially restricted states evolve through the Hilbert space spanned by the twenty lowest lying instantaneous energy eigenstates throughout the protocol.

\begin{figure}[ht]
	\begin{center}
		\includegraphics[scale=0.99]{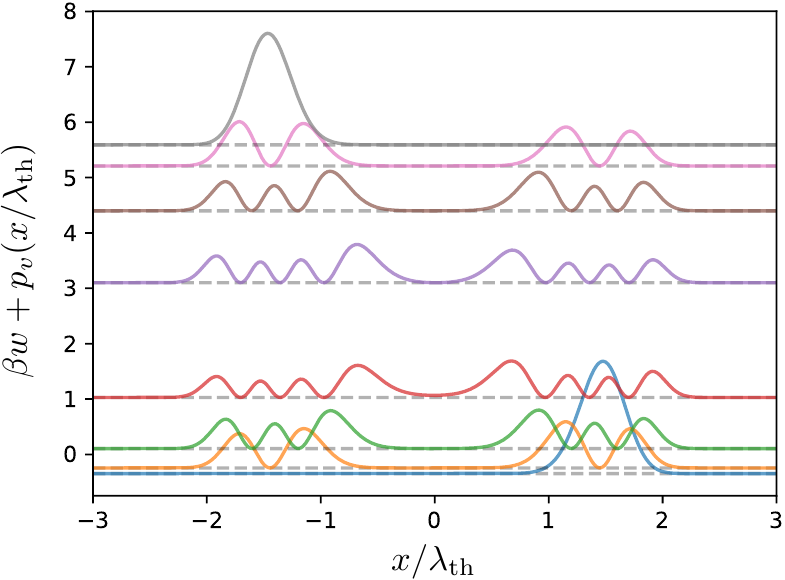}
	\end{center}
	\caption{
Spatial probability density  
		$p_v(x/\lambda_\text{th}) = \lim_{\delta x \to 0} \frac{\lambda_\text{th}}{\delta x} \, | \langle x, \, x+\delta x | v \rangle |^2$
		for the eigenstates $\ket{v}$ of the expected-work operator $\mathcal{W}$,
		offset by the corresponding work eigenvalue $w$, 
		such that $\mathcal{W} \ket{v} = w \ket{v}$.
}
	\label{fig:work_eigenstuff}
\end{figure}

Fig.~\ref{fig:work_eigenstuff}
depicts the eigenvalues and eigenstates of the expected-work operator,
via
their spatial probability density functions.
Two eigenmodes of the expected-work operator allow for work extraction,
while the other six modes require a work investment.
As one might expect, 
the ideal initial state leading to maximal work extraction is initially localized in the right well.  This state is approximately formed from equal parts of the two lowest lying energy eigenstates.
It is also intuitive that the worst input begins purely in the left well.

The remaining six non-extremal 
eigenstates of the expected-work operator are less immediately intuitive---yet they 
yield interesting insights once properly understood.
As suggested by the spatial probability density functions in
Fig.~\ref{fig:work_eigenstuff},
each work-eigenstate is approximately a linear combination of up-to-two energy eigenstates (although in detail they contain contributions from all). Some investigation revealed that each work eigenstate 
roughly corresponds to a  
relative phase combination of energy eigenstates 
that leaves the particle either mostly in the right well (for small eigenwork) or mostly in the left well (for large eigenwork) during the first half of the protocol.  
Fig.~\ref{fig:work_eigenstuff} displays some kind of symmetry in the vertical spacing of work eigenvalues that could likely be related to fluctuation relations, although that investigation remains an open opportunity.

\begin{figure}[ht]
	\begin{center}
		\includegraphics[scale=1.03]{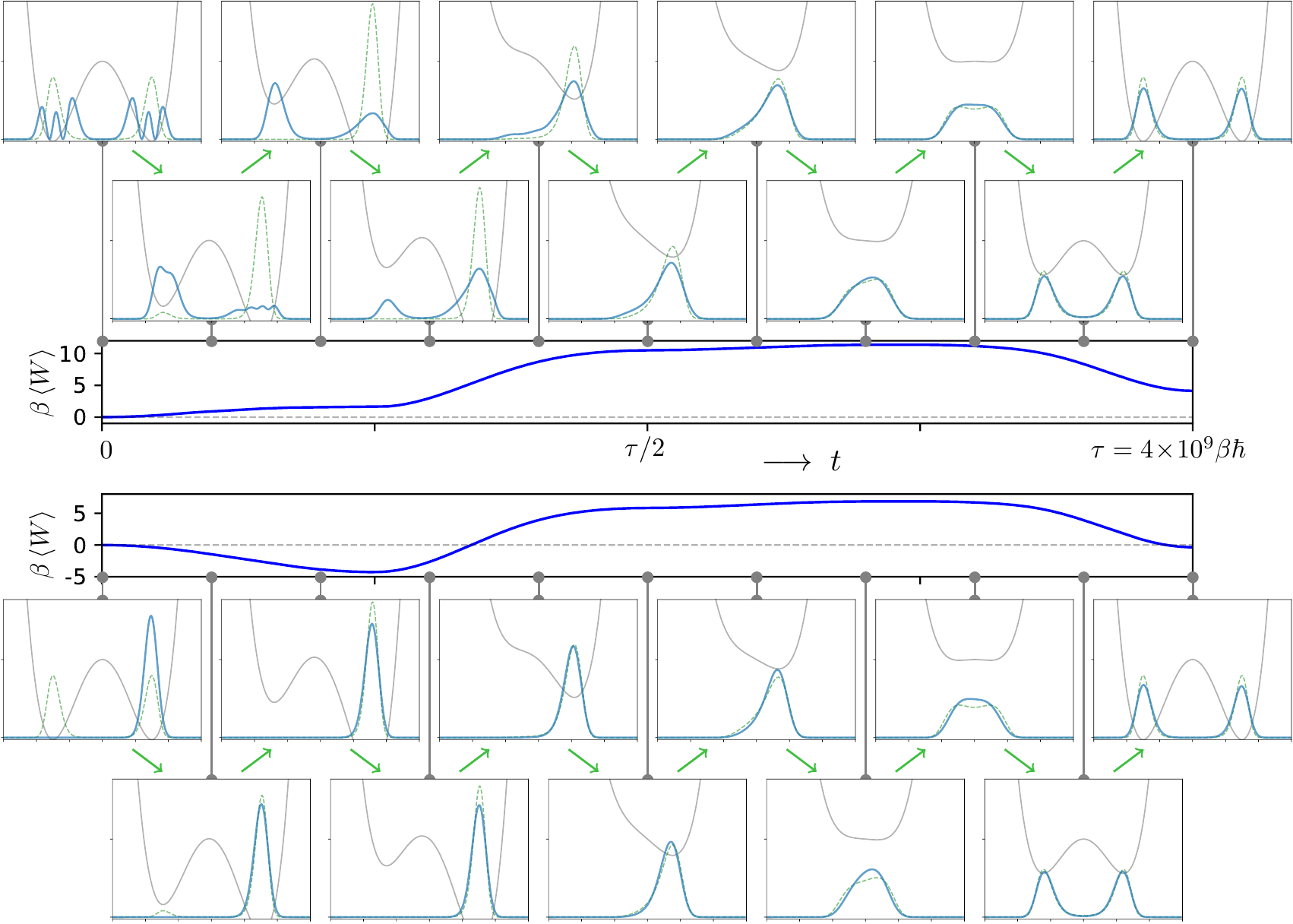}
	\end{center}
	\caption{
		Bottom panels: 
		Snapshots of spatial evolution of the ideal input to the work-extraction protocol,
		along with the resulting trajectory of expected work.
		Upper panels: 
		Snapshots of spatial evolution of a random input to the protocol,
		along with the resulting trajectory of expected work.
		Top and bottom panels show spatial probability density of the state in thick solid blue, compared to the probability density for the instantaneous equilibrium state in thin dashed green; the instantaneous potential-energy landscape, which starts and ends as a symmetric double well, is shown in thin solid gray.
		The two long middle panels show the time evolution of the expected work from each of these two initial states.
	}
	\label{fig:work_traj}
\end{figure}

The ideal input to this particular work-extraction protocol yields 
approximately $0.3472 \, \kB T$ of \emph{extracted} work
($\braket{W} = -0.3472 \, \kB T$),
which is about half of the nonequilibrium addition to free energy
($\approx \kB T \ln 2$)
available from that initial state.
This work-extraction value was initially found via the minimal eigenvalue of 
the expected-work operator, and subsequently verified via direct simulation
of this ideal input.
The spatial evolution of this ideal state, and the corresponding work trajectory are shown in the bottom half of Fig.~\ref{fig:work_traj}.
There it is apparent that the ideal state remains on the right well for the first half of the work-extraction protocol.

For comparison, the top half of Fig.~\ref{fig:work_traj} shows 
the spatial evolution and corresponding work trajectory
of a random input (in particular, we chose one of the initial energy eigenstates).
No work is extracted when operating on this state; 
rather, 
about 
$4.138 \, \kB T$ of work must be \emph{expended} in the process.
We verified that this same work value is obtained 
(through ten significant digits)
both
from the full simulation and
from the simple algebraic employment of the expected-work operator~\footnote{Notably, this initial condition was 
not used to infer	
the expected-work operator.}: 
$\tr(\rho_0 \mathcal{W}) \approx 4.138 \, \kB T$.
The series of panels 
in the top half of Fig.~\ref{fig:work_traj}
reveals the work penalty for populating the left well during the first half of the protocol.
The chronologically second snapshot 
reveals a slow local equilibration within each well.
Subsequent snapshots reveal 
an even slower tendency towards global equilibration throughout the protocol.

As anticipated by the general theory, 
the initial state leading to maximal work extraction is a pure state.
In this case, 
this state is close to, but distinct from, the input leading to minimal entropy production---the latter of which is a non-pure mixed state.
As a final note, we emphasize that we've been able to identify the ideal input within the initial low-energy subspace, although the evolving states were not limited to this subspace during their natural dynamics.

\section{Conclusion}

We have determined the ideal inputs that minimize or maximize various thermodynamic quantities for any fixed process that transforms a physical system in finite time.
Many of these optimal inputs turn out to be pure states corresponding to eigenstates of Hermitian thermodynamic operators.
We showed how to reconstruct these operators 
via observed behavior from a finite number of experimentally accessible input states.
Another class of thermodynamic quantities, based on entropies,
have mixed-state minimizers but pure-state maximizers.
The Hermitian thermodynamic operators
determine 
these ideal states too.
Our examples illustrate  
the incompatibility of common objectives:
The `ideal' input 
depends
on
whether one intends to minimize heat, minimize entropy production, maximize free-energy gain, maximize work extraction, etc.

This investigation of ideal initial states complements 
the centuries-old tradition of rather seeking ideal protocols with an assumed initial state.
Whether or not a protocol is ideal, 
our results highlight the initial-state dependence 
of a device's performance across thermodynamic metrics, 
and expose the breadth of its possible behavior.
While we emphasized thermodynamics,
the results of this paper extend easily to other domains---where the ideal inputs, as judged by some other criteria, 
like maximizing the yield of a desired output state~\cite{Halp20_Fundamental},
will be obtained from the linear operators induced by those criteria.

\section{Acknowledgements}

AK thanks Sosuke Ito for support. 
This work is supported by the Singapore Ministry of Education Tier 1 Grants RG146/20, grant no.\ FQXi-RFP-IPW-1903 (`Are quantum agents more energetically efficient at making predictions?') from the Foundational Questions Institute and Fetzer Franklin Fund (a donor-advised fund of Silicon Valley Community Foundation), the National Research Foundation, Singapore, and Agency for Science, Technology and Research (A*STAR) under its QEP2.0 programme (NRF2021-QEP2-02-P06) and the Singapore Ministry of Education Tier 2 Grant MOE-T2EP50221-0005.
This project was partly made possible through the support of Grant 62417 from the John Templeton Foundation. The opinions expressed in this publication are those of the authors and do not necessarily reflect the views of the John Templeton Foundation. 

\appendix

\section{Existence of thermodynamic operators}
\label{sec:LinearFunctionals}

Type-I expectation values---for expected work, heat, entropy flow, etc.---are linear functionals of the initial state,
since 
they can all be written as 
\begin{align}
\braket{X}_{\rho_0} =  \tr[\Xi_X(\rho_0)] ~,
\end{align}
where $\Xi_X$ is a linear superoperator acting on the 
initial density matrix of the system.
$ \tr[\Xi_X( \cdot )] $ is a linear functional 
since $\Xi_X$ and the trace operation are both linear, and the trace returns a scalar value.
The superoperator $\Xi_X$ is however relatively unwieldy,  
since it acts on the $d^2$ dimensional vector space spanned by density matrices that in turn act on a $d$-dimensional vector space $\mathcal{V}$.
The following theorem shows that 
the same linear functional can be represented via the lower-dimensional operator $\mathcal{X}$
acting on $\mathcal{V}$.
This assures the existence of our thermodynamic operators.

\begin{thm}
	For any linear 
	functional 
	$\ell$
acting on 
	a finite $d^2$-dimensional vector space 
	$\mathcal{V}_\text{big}$
	of linear operators, 
	which in turn operate on 
	a $d$-dimensional vector space $\mathcal{V}_\text{small}$,
	there exists a linear operator $\mathcal{X}$
	acting on $\mathcal{V}_\text{small}$
	such that 
\begin{align}
\ell( \cdot ) = \tr( \mathcal{X} \, \, \cdot  \, ) ~.
\end{align}
\end{thm}

\begin{proof}	
	There are only $d^2$ linearly independent linear functionals acting on 
	the $d^2$-dimensional vector space 
	$\mathcal{V}_\text{big}$.
	Let $K = \{ \ket{k} \}_{k=1}^d$ be an arbitrary orthonormal basis for 
	$\mathcal{V}_\text{small}$, with $\bra{k} = \ket{k}^\dagger$.
Note that the $d^2$  linearly independent linear functionals $\{ \tr( \, \ket{j} \bra{k} \, \cdot \, )\}_{\ket{j}, \ket{k} \in K}$ form a complete basis for the dual space of 
	$\mathcal{V}_\text{big}$.
	An arbitrary linear functional acting on $\mathcal{V}_\text{big}$ can thus be written as a 
	linear combination of these basis functionals:
\begin{align}
\ell(\cdot) = \sum_{j,k} x_{j,k} \tr( \, \ket{j} \bra{k} \, \cdot \, ) = \tr( \mathcal{X} \, \, \cdot  \, ) 
\end{align}
where 
$\mathcal{X} = \sum_{j,k} x_{j,k} \ket{j} \bra{k}$
and  $x_{j,k} = \ell(\ket{k} \bra{j})$.
\end{proof}

Hence, for type-I expectation values:
\begin{align}
	\braket{X}_{\rho_0} =  \tr[\Xi_X(\rho_0)] = 
	 \tr( \rho_0 \mathcal{X} )  ~.
\end{align}
Whenever the random variable $X$ is real-valued, its 
expectation value must also be real valued, and so $\genop$ is guaranteed to be Hermitian (see Thm.\ 2.4.3 of Hassani \cite{Hass99a}).

Sometimes the linear superoperator $\Xi_X$ can appear very complicated; but only its existence matters to guarantee the existence of the simpler operator $\mathcal{X}$.
For example, entropy flow can generically be written as 
\begin{align*}
	\braket{\EF}_{\rho_0} 
	&= 
	- \kB \int_0^\tau \tr \bigl( \dot{\rho}_t^\text{env}  \ln \stationary_{t}^\text{env} \bigr) \, dt ~,
\end{align*}	
where 
$\stationary_{t}^\text{env}$ is assumed to be independent of $\rho_0$.
If we write out $\dot{\rho}_t^\text{env} $ more explicitly,
we find that
\begin{align}
	\braket{\EF}_{\rho_0} 
	&= 
	- \kB \int_0^\tau \tr \bigl( \dot{\rho}_t^\text{env}  \ln \stationary_{t}^\text{env} \bigr) \, dt \nonumber \\
	&= 
	- \kB \int_0^\tau \tr \bigl[ ( \tfrac{d}{dt} \rho_t^\text{env} )  \ln \stationary_{t}^\text{env} \bigr] \, dt \\
	&= 
	- \kB \int_0^\tau \tr \Bigl\{ \bigl[ \tfrac{d}{dt} \tr_{\text{sys}}( U_t \rho_0 \otimes \rho_0^\text{env} U_t^\dagger ) \bigr]  \ln \stationary_{t}^\text{env} \Bigr\} \, dt \\
	&= 
 \tr \Bigl\{ - \kB \int_0^\tau \bigl[ \tfrac{d}{dt} \tr_{\text{sys}}( U_t \rho_0 \otimes \rho_0^\text{env} U_t^\dagger ) \bigr]  \ln \stationary_{t}^\text{env}  \, dt \Bigr\}  
	~. 
\end{align}	
Hence, we've found the relevant linear superoperator $\Xi_{\EF}( \cdot ) = - \kB \int_0^\tau \bigl[ \tfrac{d}{dt} \tr_{\text{sys}}( U_t \, \cdot \,  \otimes \rho_0^\text{env} U_t^\dagger ) \bigr]  \ln \stationary_{t}^\text{env}  \, dt $,
which is somewhat complicated.
No matter: 
Its mere existence proves the existence of some simpler linear operator $\boldsymbol{\Phi}$
such that 
$	\braket{\EF}_{\rho_0} = \tr( \rho_0 \boldsymbol{\Phi})$
for any initial state $\rho_0$.

A similar procedure reveals the existence of the relevant linear superoperator $\Xi( \cdot )$
for work and heat, which in turn implies the existence of the simpler expected-work operator $\mathcal{W}$ and expected-heat operator $\mathcal{Q}$.

Expanding the expression for expected work, we find
\begin{align}
	\braket{W}_{\rho_0} 
	&= 
	\int_0^\tau \tr \bigl( \rho_t \dot{H}_t  \bigr) \, dt \nonumber \\
	&= 
	\tr \Bigl[ \int_0^\tau \tr_{\text{env}}( U_t \rho_0 \otimes \rho_0^\text{env} U_t^\dagger )  \dot{H}_t  \, dt  \Bigr] ~.
	\label{eq:WorkOnJointIput}
\end{align}	
Hence, the linear superoperator
$\Xi_W( \cdot ) =  \int_0^\tau \tr_{\text{env}}( U_t \, \cdot \, \otimes \rho_0^\text{env} U_t^\dagger )  \dot{H}_t  \, dt $
implies the existence of the simpler expected-work operator $\mathcal{W}$.

Expanding the expression for expected heat, we find
\begin{align}
	\braket{Q}_{\rho_0} 
	&= 
	\int_0^\tau \tr \bigl(  \dot{\rho}_t H_t  \bigr) \, dt \nonumber \\
	&= 
	\tr \Bigl\{ 
	 \int_0^\tau \bigl[ \tfrac{d}{dt} \tr_{\text{env}}( U_t \rho_0 \otimes \rho_0^\text{env} U_t^\dagger ) \bigr]  H_t  \, dt 
	\Bigr\} ~.
\end{align}	
Hence, the linear superoperator
$\Xi_Q( \cdot ) =  \int_0^\tau \bigl[ \tfrac{d}{dt} \tr_{\text{env}}( U_t \, \cdot \, \otimes \rho_0^\text{env} U_t^\dagger ) \bigr]  H_t  \, dt $
implies the existence of the simpler expected-heat operator $\mathcal{Q}$.

\section{Thermodynamic operators characterize any measurement scheme}
\label{sec:MeasurementSchemes}

Invasive measurements change not only the probability distribution for work, but also
the expected value for work under the prescribed measurement protocol~\cite{Pera17_No}.
Our framework applies just as well to these alternative protocols with measurement interventions.
For example, let us briefly describe how our framework applies to the 
famous two-point-measurement (TPM) scheme for driven isolated quantum systems, which projects the state onto the instantaneous energy eigenbasis at both the beginning and end of the protocol~\cite{Espo09_Nonequilibrium}.

If we are interested in the TPM scheme for a driven isolated quantum system, the relevant \emph{expected-TPM-work operator} is
\begin{align}
	\mathcal{W}_\text{TPM} = \sum_{\ket{E} \in V_H} \ket{E} \! \bra{E}  \sum_{\ket{E'} \in V_{H'}} (E' - E ) \bigl| \braket{E' | U_\tau | E} \bigr|^2  ~,
\end{align}	
where $H$ and $H'$ are the initial and final Hamiltonians for the system respectively, and $U_\tau$ 
unitarily evolves 
the system between the two projective measurements.
In the TPM scheme,
the system transitions from initial energy eigenstate $\ket{E}$ 
to final energy eigenstate $\ket{E'}$
with probability 
$\braket{E | \rho_0 | E} \bigl| \braket{E' | U_\tau | E} \bigr|^2  $,
resulting in 
a TPM work value of $E'-E$.
It is easy to verify that
\begin{align}
\braket{W_\text{TPM}}_{\rho_0} = 
\sum_w
	w \Pr(W_\text{TPM} = w) 
	= \tr( \rho_0 \mathcal{W}_\text{TPM}) ~,
	\label{eq:TPM_RV_and_operator}
\end{align}
where $W_\text{TPM}$ is the random variable for
the TPM work outcome, and the sum runs over all $w \in \{ E' - E : E \in \Lambda_H , E' \in \Lambda_{H'} \}$.

Ref.~\cite{Pera17_No} shows that $\mathcal{W} \neq \mathcal{W}_\text{TPM} $.
Invasive measurements can change the expectation value for work 
(and indeed for other thermodynamic variables too).

If desired, our framework can be employed to find the ideal inputs leading to minimal or maximal TPM work, from
the minimal and maximal eigenstates of  
$\mathcal{W}_\text{TPM}$.
Indeed,
for any measurement scheme---including also the one-time-measurement scheme~\cite{Deff16_Quantum, Beye20_Work, Sone20_Quantum} or any other---thermodynamic operators can be constructed, and our framework can be applied to identify both the breadth of behavior and the ideal inputs within the scheme.

\section{Expectation-value invariance to subensemble decomposition}
\label{sec:RVsForQuantumThermo}

Our contributions in the main text exclusively involve
expectation values, so are 
invariant to density-matrix decompositions.
Here we show this invariance explicitly.
As an added benefit,
it can be useful to see how the
quantum random variables 
induced by a subensemble decomposition
relate to and generalize their classical counterparts.

In the classical case, fluctuations are primarily due to subjective uncertainty about which state the environment is in (although uncertainty of the system state also plays a role).
Since quantum results should include classical results as a special case,
we must also consider the role of subjective uncertainty in the quantum case.
In general, this can be addressed via a probability distribution over (possibly nonorthogonal) subensembles of the joint system--environment supersystem~\cite{Alla05_Fluctuations}.

There are many ways to decompose the initial joint density matrix of the system--environment supersystem. 
In fact for a non-pure mixed state $\rho_0^\text{tot}$,
there are infinitely many pairs $\bigl( (p_n)_n, (\sigma_n)_n \bigr)$
of probability distributions $(p_n)_n$ 
over constituent density matrices $(\sigma_n)_n$,
for which 
$\rho_0^\text{tot} = \sum_n \sigma_n  p_n$.
Each of these decompositions can represent a physically relevant preparation of the supersystem~\cite{Alla05_Fluctuations, Anza21_Beyond}.

Each subensemble decomposition $\bigl( (p_n)_n, (\sigma_n)_n \bigr)$
induces a random variable 
$X\bigl( (p_n)_n, (\sigma_n)_n \bigr)$.~\footnote{This generalizes the classical case,
where 	
the 
most fine-grained version of the
random variable $X$ is uniquely fixed by the probability distribution in the classical basis.
Classical subensemble decompositions include coarse grainings,
but also include 
collections of initial distributions with overlapping support.
The quantum case furthermore allows for basis freedom.  
In each case, the physical relevance
of subensemble decompositions
correspond to possible physical preparations of the system.
}   
With probability $p_n$, the subensemble $\sigma_n$ will be realized, in which case the random variable $X$ takes on 
the value $x(\sigma_n)$.
When the thermodynamic quantity is a linear functional of the initial state---as is the case for work, heat, entropy flow, etc.---it's easy to see that 
\begin{align}
	\Bigl\langle X\bigl( (p_n)_n, (\sigma_n)_n \bigr) \Bigr\rangle_n
	= \sum_n x(\sigma_n) p_n
	= x \Bigl( \sum_n \sigma_n p_n \Bigr)
	= x(\rho_0^\text{tot}) ~,
\end{align}	
independent of the decomposition.
In the main text, we always consider $\rho_0^\text{tot} = \rho_0 \otimes \rho_0^\text{env}$ for a fixed initial state of the environment $\rho_0^\text{env}$.
For these thermodynamic
linear functionals, 
we thus use the shorthand notation \begin{align}
	\braket{X}_{\rho_0} = x(  \rho_0 \otimes \rho_0^\text{env} )  ~.
\end{align}

As a consequence, the expectation value
$\braket{X}_{\rho_0}$ 
and
thermodynamic operator $\mathcal{X}$
are both 
invariant to subensemble decompositions.

Let's take thermodynamic work (without measurement interventions)
as an example.
From Eq.~\eqref{eq:WorkOnJointIput}, we see that work can be quantified for any 
initial joint state $\sigma$ of the supersystem 
via 
\begin{align}
w(\sigma)= \tr \Bigl[ \int_0^\tau \tr_{\text{env}}( U_t \sigma U_t^\dagger )  \dot{H}_t  \, dt  \Bigr] ~.
\end{align}

If the system and environment are unmeasured, then it is natural 
for 
eigenvalues of the uncorrelated density matrix to represent subjective uncertainty. 
We can thus take 
$\bigl( (p_n)_n, \ket{n} \! \bra{n} \bigr)$ as our subensemble decomposition,
where $\ket{n}$ is an eigenstate of $\rho_0^\text{tot}$.
In this case, work will be a random variable $W$
which takes on the value $w^{(n)} = w\bigl( \ket{n} \! \bra{n}  \bigr)$
with probability
$\sum_m p_m \delta_{w^{(m)} , w^{(n)} }$.
We can then write
\begin{align}
\braket{W}_{\rho_0} 
= \sum_n  w^{(n)} p_n 
= w(\rho_0 \otimes \rho_0^\text{env})
= \tr( \rho_0 \mathcal{W}) ~.
\label{eq:QuantumRV_and_operator}
\end{align}
We see that the expectation value and the thermodynamic operator are both independent of the chosen subensemble decomposition.

\section{Operator expressions for type-II expectation values}
\label{sec:OperatorExpressionForType2}

The expectation values of entropy production, the reduction in nonequilibrium free energy, and the change in von Neumann entropy
can all be written as:
\begin{align}
	f_{\rho_0}^{(\genop)}
	&= \tr(\rho_0 \genop) + S(\rho_\tau) - S(\rho_0)\\
	&= \tr(\rho_0 \ln \rho_0) - \tr(\rho_0 \ln e^{-\genop})  +  S(\rho_\tau) \\
	&= \tr(\rho_0 \ln \rho_0) - \tr\{\rho_0 \ln [e^{-\genop} / \tr(e^{-\genop})] \}  - \ln [ \tr(e^{-\genop})] +  S(\rho_\tau) \\
	&=
	\text{D}[\rho_0 \| \omega] - \ln[ \tr(e^{-\genop}) ] + S(\rho_\tau) ~,
\end{align}
where $\omega \coloneqq e^{-\genop} / \tr(e^{-\genop})$.

For reliable reset processes, for which $\rho_\tau \approx r_\tau$ is very nearly independent of the input,
it is clear that $\omega$ minimizes $f_{\rho_0}^{(\genop)}$,
where it takes on the minimal value 
$f_{\omega}^{(\genop)} = S(r_\tau) - \ln[ \tr(e^{-\genop}) ]$.

Ref.~\cite{Kolc21_State} 
pointed out that entropy production can be written in this form.
Here we notice that this type of relation applies also to a larger family of thermodynamic quantities.

\section{Generalized Gell-Mann matrices as a standard operator basis}
\label{sec:GGM}

For concreteness, we will describe an operator basis that can be used in any finite dimension $d$.

Let 
$\bigl\{ \ket{1} , \dots , \ket{d} \bigr\}$
be an orthonormal basis for 
the $d$-dimensional Hilbert space of our system under study.
We can then construct the generalized Gell-Mann matrices to complete a matrix basis for density matrices acting on this Hilbert space.

Let $n \in \mathbb{R}$ be any convenient constant.
$\vec{\Gamma}$ will contain
$d-1$ diagonal matrices
\begin{align}
n 
	\Bigl( \tfrac{2}{\ell^2 + \ell} \Bigr)^{\! 1/2}
\,
	\Bigl[ \Bigl( \sum_{j=1}^\ell \ket{j} \bra{j} \Bigr) - \ell \ket{\ell+1} \bra{\ell+1} \Bigr] ~,
\end{align}
with $\ell \in \{ 1, \dots , d-1 \}$.	
It
will contain
$(d^2 - d)/2$
distinct non-diagonal symmetric matrices
\begin{align}
n \bigl( \ket{k}\bra{j} + \ket{j}\bra{k}  \bigr) ~,
\end{align}	
with $1 \leq j < k \leq d$.
It will also contain
$(d^2 - d)/2$
distinct antisymmetric matrices
\begin{align}
	i n \bigl(  \ket{k}\bra{j} - \ket{j}\bra{k} \bigr) ~,
\end{align}	
with $1 \leq j < k \leq d$.
The ordering of these $d^2-1$ matrices is arbitrary.

Notice that these operators are traceless and satisfy 
$\tr(\Gamma_m \Gamma_n) = 2 n^2 \delta_{m,n}$.
Hence, 
$\constname = 2 n^2$.
If we choose $n = \sqrt{\tfrac{d - 1}{2d}}$,
then $\constname = \const$,
pure states will always have a generalized Bloch vector of unit length,\footnote{This can be verified through the expression $\tr(\rho^2) = \tr \bigl[ (I/d + \vec{b} \cdot \vec{\Gamma}) (I/d + \vec{b} \cdot \vec{\Gamma}) \bigr]$,
from which we find	
$b = \sqrt{\frac{ \tr(\rho^2) d - 1}{ d-1 }}$.} 
and the generalized Bloch vector will reduce to the standard Bloch vector in $d=2$.

\section{Composite operator bases}
\label{sec:Composite_bases}

Suppose we have normalized Hermitian operator bases, ($\ident_d / d , \, \vec{\Gamma}$) for a $d$-dimensional subsystem,
and ($\ident_{d'} / d', \, \vec{\Gamma'}$) for a
$d'$-dimensional subsystem.
The operator bases satisfy
\begin{align}
	\tr(\Gamma_n) = 0 \quad \text{and} \quad
	\tr(\Gamma_m \Gamma_n) = \constname \delta_{m,n}
\end{align}	
and 
\begin{align}
	\tr(\Gamma_n') = 0 \quad \text{and} \quad
	\tr(\Gamma_m' \Gamma_n') = \constname' \delta_{m,n} ~.
\end{align}	

We can then construct a normalized Hermitian basis for operators acting on the composite $d''$-dimensional Hilbert space, where $d''=dd'$:
The new operator basis,
($\ident_{d''}/d'' , \, \vec{\Gamma''}$)
satisfies 
\begin{align}
	\tr(\Gamma_n'') = 0 \quad \text{and} \quad
	\tr(\Gamma_m'' \Gamma_n'') = \constname'' \delta_{m,n} ~,
	\label{eq:CompositeProperties}
\end{align}	
where the new composite operator basis contains
\begin{align}
\vec{\Gamma''} = \Bigl( 
\sqrt{\tfrac{\constname''}{\constname' d}} \ident_d \otimes \vec{\Gamma'} , \, \, \,
\sqrt{\tfrac{\constname''}{\constname d'}} \vec{\Gamma''} \otimes \ident_{d'}  , \, \, \,
\sqrt{\tfrac{\constname''}{\constname \constname'}} \vec{\Gamma} \otimes \vec{\Gamma'} 
\Bigr)	~.
\end{align}	
Eq.~\eqref{eq:CompositeProperties} 
can be verified
via the identities
$(A \otimes B) (C \otimes D) = (AC) \otimes(BD)$
and 
$\tr(A \otimes B) = \tr(A) \tr(B)$.

This allows us, for example, to build up a normalized Hermitian operator basis for
\begin{itemize}
	\item
many qubits, 
using tensor products of the local Pauli operators,
or 
\item
a qubit and qutrit, using tensor products of their local Pauli and Gell-Mann operators, 
\item
etc.
\end{itemize}

\section{Existence of, and expression for, the entropy-flow vector}
\label{sec:EFvec}

Leveraging the decomposition of the system state in terms of its generalized Bloch vector,
$\rho_t = \sysref + \vec{b}_t \cdot \vec{\Gamma}$,
we can expand the general expression for entropy flow
given in Ref.~\cite{Riec21_Impossibility} and discussed in the main text,
to express entropy flow in terms of the entropy-flow vector and initial Bloch vector:
\begin{align}
\braket{\EF}_{\rho_0} 
&= 
- \kB \int_0^\tau \tr \bigl( \dot{\rho}_t^\text{env}  \ln \stationary_{t}^\text{env} \bigr) \, dt \\
&= 
- \kB \int_0^\tau \tr \biggl\{  \Bigl[ \tfrac{d}{dt} \tr_\text{sys} \bigl( U_{0:t} \rho_0 \otimes \rho_0^\text{env} U_{0:t}^\dagger \bigr) \Bigr]
   \ln \stationary_{t}^\text{env} \biggr\} \, dt \\
&= 
- \kB \int_0^\tau \tr \biggl\{  \Bigl[ \tfrac{d}{dt} \tr_\text{sys} \bigl( U_{0:t} (\sysref + \vec{b}_0 \cdot \vec{\Gamma}) \otimes \rho_0^\text{env} U_{0:t}^\dagger \bigr) \Bigr]
\ln \stationary_{t}^\text{env} \biggr\} \, dt \\   
&= 
\underbrace{- \kB \int_0^\tau \tr \biggl\{  \Bigl[ \tfrac{d}{dt} \tr_\text{sys} \bigl( U_{0:t} \sysref \otimes \rho_0^\text{env} U_{0:t}^\dagger \bigr) \Bigr]
\ln \stationary_{t}^\text{env} \biggr\} \, dt }_{= \braket{\EF}_{\sysref}}
+ 
\vec{b}_0 \cdot 
\Biggl(
\underbrace{
 - \kB \int_0^\tau \tr \biggl\{  \Bigl[ \tfrac{d}{dt} \tr_\text{sys} \bigl( U_{0:t}  \vec{\Gamma} \otimes \rho_0^\text{env} U_{0:t}^\dagger \bigr) \Bigr]
\ln \stationary_{t}^\text{env} \biggr\} \, dt }_{\eqqcolon \vec{\varphi}}  \Biggr)
 \\  
&= \braket{\EF}_{\sysref} + \vec{b}_0 \cdot \vec{\varphi} ~.
\end{align}

Other thermodynamic quantities can be decomposed similarly.

\section{Expression for the second derivative of von Neumann entropy}
\label{sec:SecondDerivativeOfEntropy}

The second derivative of von Neumann entropy, with respect to the elements of the Bloch vector of the quantum state,
can be calculated via the spectral decomposition of the quantum state
$\pbb = \sum_{k}\lambda_{k} \ket{k} \bra{k}$ and its logarithm
$\ln \pbb = \sum_{k} \ln (\lambda_{k} ) \ket{k} \bra{k}$.
Assuming 
non-degeneracy of the eigenvalues, 
we find
\begin{align}
	\partial_{b_{m}} \partial_{b_{n}} S [\pbb]
	&=
	- \partial_{b_{m}} 
	 \tr \bigl[ \Gamma_{n} \ln \pbb \bigr]  \\
	 &= -\tr \bigl[ \Gamma_{n} \, \partial_{b_{m}} \ln \pbb \bigr] ~,
\end{align}
where
\begin{align}
 \partial_{b_{m}} \ln \pbb 
 &=
 \sum_k \biggl\{ 
 \frac{1}{\lambda_k} \underbrace{( \partial_{b_{m}} \lambda_k ) }_{ \braket{k | \Gamma_m | k}}
 \ket{k} \bra{k} + \ln (\lambda_k) \Bigl[ 
 \!\!\!\!
 \underbrace{ \bigl( \partial_{b_{m}} \ket{ k} \bigr) }_{\sum_{\ell \neq k} \frac{\braket{\ell| \Gamma_m| k} }{ \lambda_{k} - \lambda_{\ell}} \ket{\ell} }  \!\!\!\!\!\!  \bra{ k } \,
 + \,
 \ket{k} 
  \!\!\!\!\!\!
 \underbrace{ \bigl( \partial_{b_{m}} \bra{ k} \bigr)}_{\sum_{\ell \neq k} \frac{\braket{ k | \Gamma_m| \ell } }{ \lambda_{k} - \lambda_{\ell}} \bra{\ell} }  \!\!\!\!
\Bigr]
  \biggr\}  ~.
\end{align}

Hence,
\begin{align}
	\partial_{b_{m}} \partial_{b_{n}} S [\pbb]
	&= 
	\biggl(  \sum_k \tfrac{1}{\lambda_k} 
	\braket{k | \Gamma_n | k} \braket{ k | \Gamma_m | k} \biggr)
	+
	\biggl(  \sum_{k, \ell \atop \ell \neq k} \frac{ \ln \lambda_k }{\lambda_k - \lambda_\ell} 
	\braket{k | \Gamma_n | \ell} \braket{ \ell | \Gamma_m | k} \biggr)
	+
	\biggl(  \sum_{k, \ell \atop \ell \neq k} \frac{ \ln \lambda_k }{\lambda_k - \lambda_\ell} 
	\braket{\ell | \Gamma_n | k} \braket{ k | \Gamma_m | \ell} \biggr) ~.
	\label{eq:SecondDerivativeOfEntropy}
\end{align}
Upon swapping the name of the indices in the last of the three terms on the right-hand side of Eq.~\eqref{eq:SecondDerivativeOfEntropy}, we find
\begin{align}
 \sum_{k, \ell \atop \ell \neq k} \frac{ \ln \lambda_k }{\lambda_k - \lambda_\ell} 
\braket{\ell | \Gamma_n | k} \braket{ k | \Gamma_m | \ell} 
&=
 \sum_{\ell, k \atop k \neq \ell} \frac{ \ln \lambda_\ell }{\lambda_\ell - \lambda_k} 
\braket{k | \Gamma_n | \ell} \braket{ \ell | \Gamma_m | k} \\
&=
 \sum_{k, \ell \atop \ell \neq k}
 \frac{ - \ln \lambda_\ell }{\lambda_k - \lambda_\ell} 
\braket{k | \Gamma_n | \ell} \braket{ \ell | \Gamma_m | k} 
~.
\end{align}	
Noting that 
$\lim_{b \to a} \frac{\ln a-\ln b}{a-b} = 1/a$,
we can thus combine all three terms 
on the right-hand side of Eq.~\eqref{eq:SecondDerivativeOfEntropy} 
to find 
\begin{align}
	\partial_{b_{m}} \partial_{b_{n}} S [\pbb]
	&=
	- \sum_{k,\ell} \phi(\lambda_{k},\lambda_{\ell}) \braket{ k \vert \Gamma_{n} \vert \ell } \braket{  \ell \vert \Gamma_{m} \vert k }
	~,
\end{align}
where 
we defined $\phi(a,b) \coloneqq \frac{\ln a-\ln b}{a-b}$ (with $\phi(a,a)=1 / a$
by continuity).
Note that $\phi(a,b)$ is the reciprocal of the logarithmic mean of the eigenvalues $a$ and $b$. 

Similarly, 
using the spectral decomposition
of $\pbbF = \sum_{k} \lambda_{k}' \ket{ k' } \bra{k'}$,
we find 
\begin{align}
	\partial_{b_{m}} \partial_{b_{n}} S [\pbbF]
	&=
	- \sum_{k,\ell} \phi(\lambda_{k}',\lambda_{\ell}') \braket{ k' \vert \Gamma_{n}' \vert \ell' } \braket{  \ell' \vert \Gamma_{m}' \vert k' }
	~.
\end{align}

\section{Dynamical equations for a qubit-reset device}
\label{sec:EOMforReset}

Here we review the dynamical equations for a particular device that 
implements
qubit reset, used in our first example:   
As in Refs.~\cite{Mill20, Riec21_Impossibility},
the device works 
by changing both the energy gap and spatial orientation of the energy eigenstates of the qubit while it is in weak contact with a thermal reservoir at inverse temperature $\beta = 1/(\kB T)$.
Over the finite-time protocol, 
from time $0$ through $\tau = 50 \beta \hbar$,
the time-varying Hamiltonian is
\begin{align}
	H_{t} = \frac{E_t}{2} \bigl[ \cos (\theta_t) \sigma_z + \sin (\theta_t) \sigma_x \bigr] ~,
\end{align}
with 
$E_t = \kB T [1 + 49 \sin^2 \bigl( \tfrac{\pi t }{ 100 \beta \hbar } \bigr)] / 5$
and 
$\theta_t = \pi t / (50 \beta \hbar)$.
While $E_t$ quantifies the energy gap between the system's instantaneous energy eigenstates,
$\theta_t$ parametrizes the instantaneous orientation of the energy eigenbasis relative to the `computational' $z$-basis.
The dynamics are
well-described by a time-dependent quantum master equation
\begin{align}
	\dot{\rho}_t = \mathcal{L}_{x_t} (\rho_t)
	= \frac{i}{\hbar} [\rho_t, H_{x_t}] 
	& + \frac{c E_t}{\hbar} (N_{x_t} + 1) \mathcal{D}[L_{x_t}] (\rho_t) 
+ \frac{c E_t}{\hbar} N_{x_t} \mathcal{D}[L_{x_t}^\dagger] (\rho_t)
\end{align} 
where $\mathcal{D}[L](\rho) = L \rho L^\dagger - \tfrac{1}{2} \{ L^\dagger L, \rho \}$,
$N_{x_t} = (e^{\beta E_t} - 1)^{-1}$,
and $c = 1/5$ is the coupling strength to the bath.
The time-dependent lowering operator can be represented as
\begin{align}
	L_{x_t} = \frac{1}{2} 
	\bigl[
	\cos (\theta_t) \sigma_x
	- i \sigma_y
	- \sin (\theta_t) \sigma_z  
	\bigr]
\end{align}
and satisfies the detailed balance condition 
$[L_{x_t}, H_{x_t}] = E_t L_{x_t}$ ~\cite{Mill20, Manz15}.
Transitions thus occur between
instantaneous energy eigenstates of the system.
In particular, $\mathcal{D}[L_{x_t}] (\rho_t)$
takes the excited population and shifts it to the ground state,
while $ \mathcal{D}[L_{x_t}^\dagger] (\rho_t)$
takes the ground-state population and shifts it to the excited state.
Moreover, the ratio of transition rates between the excited and ground-state populations satisfies detailed balance since 
$(N_{x_t} + 1) / N_{x_t} = e^{\beta E_t}$.

\section{Time-dependent double-well dynamics}
\label{sec:DoubleWellDynamics}

 Here we construct a simple model for the dynamics of a non-relativistic charged
 particle in a time-dependent double-well of potential energy across one spatial dimension while it is immersed in a bosonic bath at temperature $T$.
The protocol is 
capable of extracting work from some
states that are initially localized in the right well.

In particular, we simulate a work-extraction protocol via the time-dependent system Hamiltonian:
\begin{align}
	H_t = \frac{p_x^2}{2m_q} + V_t
\end{align}	
where 
$p_x = -i\hbar \partial_x$ is the momentum operator and 
$m_q$ is the mass of the system with charge $q$.
The time-dependent potential-energy landscape is
\begin{align}
	V_t = 
	16 h_0 \Bigl( \frac{x}{w_0} \Bigr)^4 
	- 8  h_0 \Bigl( \frac{x}{w_0} \Bigr)^2 g_t 
	- h_0 \frac{x}{w_0}  f_t
	~, 
\end{align}
where 
$\barrier_t$ and
$\tilt_t$
are non-negative scalar functions of time,
with $\barrier_0 = \barrier_\tau = 1$
and 
$\tilt_0 = \tilt_\tau = 0$.
To enter an interesting thermodynamic regime, we choose
the initial energy barrier between the two wells to be $h_0 = 8 \kB T$,
where $\kB T = 1/\beta$ is the thermal energy of the environment.
To enter an interesting quantum regime, we choose
the initial separation between the bottom of the two wells to be $w_0 = 3 \lambda_\text{th}$, where 
$\lambda_\text{th} = \hbar \sqrt{2 \pi \beta / m_q}$
is the thermal de Broglie wavelength.
To induce nonequilibrium quantum dynamics,
we choose a sufficiently fast protocol.

$\barrier_t$ controls the height of the barrier
while $\tilt_t$ controls the tilt of the energetic landscape throughout the protocol.
During the protocol, the potential is tilted to the right, the barrier is lowered, the potential is untilted, and the barrier is then reintroduced, in that order.
The exact form of the control protocol is given by
\begin{align}
	g_t = 
	\begin{cases}
		1 			& \text{if } 0 \leq t < \tau/4 \\
		\sin^2(2 \pi t / \tau) & \text{if } \tau/4 \leq t < \tau/2 \\
		0 & \text{if } \tau/2 \leq t < 3 \tau/4 \\
		\cos^2(2 \pi t / \tau) & \text{if } 3 \tau/4 \leq t \leq \tau 
	\end{cases}	
\end{align}	
and
\begin{align}
	f_t / f^\text{max} = 
	\begin{cases}
		\sin^2(2 \pi t / \tau)	& \text{if } 0 \leq t < \tau/4 \\
		1 & \text{if } \tau/4 \leq t < \tau/2 \\
		\cos^2(2 \pi t / \tau) & \text{if } \tau/2 \leq t < 3 \tau/4 \\
		0 & \text{if } 3 \tau/4 \leq t \leq \tau 
	\end{cases}	
~.
\end{align}	
We chose the maximal tilt to be 
$f^\text{max} = 8/(3 \sqrt{2 \pi})$.

The system Hamiltonian $H_t = p_x^2 / 2m_q + V_t$ has a countably infinite orthonormal set of instantaneous energy eigenstates $\{ \ket{E_t^{(n)}} \}_n$ with corresponding instantaneous energy eigenvalues $\{ E_t^{(n)} \}_n$, ordered such that $E_t^{(n)}  \geq E_t^{(m)} $ if $n > m$.
Hence, $H_t \ket{E_t^{(n)}} = E_t^{(n)} \ket{E_t^{(n)}} $,
and the Hamiltonian has the simple eigen-representation
$H_t = \sum_{n = 1}^{\infty} E_t^{(n)} \ket{E_t^{(n)}} \! \bra{E_t^{(n)}} $.
It is useful to represent differences among eigen-energies via the relevant angular frequencies $\omega_t^{(n,m)}$, such that $\hbar \omega_t^{(n,m)} = E_t^{(n)} - E_t^{(m)}$.

The system is immersed in a photon bath throughout the protocol.
Photons in the environment can induce transitions among the instantaneous energy eigenstates of the system, with 
absorption and stimulated-emission 
transition rates proportional to the number of photons with the relevant transition energy $E_t^{(n)} - E_t^{(m)}$.  
In general, this expected number of bath photons, per mode with frequency $\omega / 2 \pi$, can be a function of time $\braket{N^{(\omega)}}$, which is related to the intensity.

Suppose for now that the energy eigenstates of the system are non-degenerate.~\footnote{Since the energy eigenstates are non-degenerate almost everywhere along the continuous time interval from 0 to $\tau$,
we will not worry about the potential complications of choosing a preferred basis for transitions to or from an energy-degenerate subspace.}
Let $L_t^{(m,n)} = \ket{E_t^{(m)}} \bra{E_t^{(n)}}$.
If $E_t^{(n)} > E_t^{(m)}$, this can be interpreted as 
the lowering operator between these two instantaneous energy eigenstates.
It is useful to reflect on the physical implications of the mathematical dissipator $\mathcal{D}[L ] (\rho) = L \rho L^\dagger - \tfrac{1}{2} \{ L^\dagger L , \rho \}$.
Notice that the 
dissipator 
\begin{align}
\mathcal{D}[L_t^{(m,n)} ] (\rho_t) =  \ket{E_t^{(m)}} \bra{E_t^{(n)}} \rho_t  \ket{E_t^{(n)}} \bra{E_t^{(m)}}
- \tfrac{1}{2}  \ket{E_t^{(n)}} \bra{E_t^{(n)}} \rho_t
- \tfrac{1}{2}  \rho_t \ket{E_t^{(n)}} \bra{E_t^{(n)}} 
\end{align}
fully removes the $\ket{E_t^{(n)}} $ population
and shifts it down to $\ket{E_t^{(m)}} $.
Analogously, 
$\mathcal{D}[(L_t^{(m,n)})^\dagger ] (\rho_t) = \mathcal{D}[L_t^{(n,m)} ] (\rho_t) $
fully removes the $\ket{E_t^{(m)}} $ population
and shifts it up to $\ket{E_t^{(n)}} $.
The \emph{rate} at which these two processes happen can be denoted by $r_t^{(n \to m)}$ and $r_t^{(m \to n)}$ respectively.
This can then be integrated into a quantum master equation that takes on the Linblad form:
\begin{align}
	\dot{\rho}_t = \mathcal{L}_{t} (\rho_t)
	= \frac{i}{\hbar} [\rho_t, H_{t}] 
	& + 
	\sum_n \sum_{m < n}
	r_t^{(n \to m)}  \mathcal{D}[L_t^{(m,n)} ] (\rho_t) 
	+
	r_t^{(m \to n)}  \mathcal{D}[L_t^{(n,m)} ] (\rho_t) ~.
	\label{eq:Generic_Lindbladian}
\end{align} 

As already mentioned, the rate of excitation $r_t^{(m \to n)} $ will be proportional to the expected number of bath particles $\braket{N^{(\omega)}}$ with the relevant energy $\hbar \omega_t^{(n, m)}$.
The rate of emission is however more subtle, since it involves both stimulated and spontaneous emission.

A standard quantum electrodynamic calculation (introduced a century ago by Dirac~\cite{Dirac27_Quantum}) 
relates absorption and emission between any two energy eigenstates of a system, given 
the photon intensity of the environment.
In particular, following 
Fermi's golden rule for transitions on the joint state space of system and photons 
~\cite[\S 1.3]{Schwartz14_Quantum}, for $E_t^{(m)} < E_t^{(n)}$, the rate of absorption is given by
\begin{align}
	r_t^{(m \to n)} = \frac{2 \pi}{\hbar} | \mathcal{M}_t^{(m,n)} |^2 \braket{N^{(\omega_t^{(n,m)})}}
	\label{eq:GenAbsorptionRate}
\end{align}	   
while the net rate of both stimulated and spontaneous emission
is given by
\begin{align}
	r_t^{(n \to m)} = \frac{2 \pi}{\hbar} | \mathcal{M}_t^{(m,n)} |^2 \bigl( \braket{N^{(\omega_t^{(n,m)})}} + 1 \bigr) ~.
	\label{eq:GenEmissionRate}
\end{align}	   
Here, $\mathcal{M}_t^{(m,n)}$ is the transition amplitude between 
$\ket{E_t^{(m)} }$ and $\ket{E_t^{(n)} }$
induced by the background radiation.
We use the standard interaction Hamiltonian $H_I = \frac{-q}{m_q c } A_x \otimes p_x$ 
on the joint state space of photons and system,
where 
$c$ is the speed of light, and
$A_x $ is the $x$-component of polarization 
of the
quantum field for the electromagnetic vector potential $\vec{A}(\vec{r}, t)$.
For simplicity, we assume the system is some distance $z$ from the photon source,
so that the wavevector of radiation is $\vec{k} = k \hat{z} = (\omega / c) \hat{z} \perp \hat{x}$.
In the standard electric-dipole approximation, 
the absorption cross section---i.e., (power absorbed by the $m \to n$ transition) / (incident power per area)---is given by
$\sigma_\text{abs} = 4 \pi^2 (q/\text{e})^2 \alpha \omega \bigl| \braket{E_t^{(m)} | x | E_t^{(n)} } \bigr|^2 \delta(\omega - \omega_t^{(n,m)})$,
where e is the charge of a single electron, and $\alpha \approx 1/137$ is the fine-structure constant~\cite{Saku14_Modern}.
Hence the transition rate of absorption is related to the spectral intensity (power per area per angular frequency) $\mathcal{I}(\omega)$
via 
\begin{align}
	r_t^{(m \to n)} 
	&= 
	\int_\omega
	\frac{\mathcal{I}(\omega) \sigma_\text{abs}}{\hbar \omega}
	\, d \omega
	\\
	&= 
	\frac{4 \pi^2 \alpha}{\hbar}   \Bigl( \frac{q}{\text{e}} \Bigr)^2
	\bigl| \braket{E_t^{(m)} | x | E_t^{(n)} } \bigr|^2 \, \mathcal{I}(\omega_t^{(n,m)})  ~.
\end{align}	   
Arbitrarily far from equilibrium,
intensity and expected occupation are 
related by 
$\mathcal{I}(\omega) = \frac{\hbar \omega^3}{\pi^2 c^2} \braket{N^{(\omega)}}$.
Incorporating this, we find
\begin{align}
	r_t^{(m \to n)} 
	&= 
	\frac{4 \alpha }{c^2} 
	(q/\text{e})^2
	(\omega_t^{(n,m)})^3  
	\bigl| \braket{E_t^{(m)} | x | E_t^{(n)} } \bigr|^2 \, \braket{N^{(\omega_t^{(n,m)})}} \\
&= 
	\frac{\kB T}{\hbar}
	8 \pi \alpha (q/\text{e})^2  
	\frac{\kB T}{m_q c^2}
	\Bigl( \frac{\hbar \omega_t^{(n,m)} }{ \kB T } \Bigr)^3
	\bigl| \braket{E_t^{(m)} | (x / \lambda_\text{th} ) | E_t^{(n)} } \bigr|^2 \, \braket{N^{(\omega_t^{(n,m)})}}  
	\label{eq:AbsorptionRate_final}
\end{align}	   
for $n>m$.
Comparing with Eqs.~\eqref{eq:GenAbsorptionRate} and \eqref{eq:GenEmissionRate},
this also implies 
\begin{align}
	r_t^{(n \to m)} 
	&= 
	\frac{\kB T}{\hbar}
	8 \pi \alpha (q/\text{e})^2  
	\frac{\kB T}{m_q c^2}
	\Bigl( \frac{\hbar \omega_t^{(n,m)} }{ \kB T } \Bigr)^3
	\bigl| \braket{E_t^{(m)} | (x / \lambda_\text{th} ) | E_t^{(n)} } \bigr|^2 \, \bigl( \braket{N^{(\omega_t^{(n,m)})}}  + 1 \bigr)	
	\label{eq:EmmissionRate_final}
\end{align}	   
for $n>m$.

The dynamics can thus be expressed as \begin{align}
	\dot{\rho}_t 
	= \frac{i}{\hbar} [\rho_t, H_{t}] 
	& + \gamma
	\sum_n \!\!
	\sum_{m < n} \!
	\Bigl( \frac{\hbar \omega_t^{(n,m)} }{ \kB T } \Bigr)^3
	\bigl| \braket{E_t^{(m)} | \frac{x}{\lambda_\text{th} } | E_t^{(n)} } \bigr|^2 \Bigl\{ \!
	\bigl( \braket{N^{(\omega_t^{(n,m)})}} + 1 \bigr) \mathcal{D}[L_t^{(m,n)} ] (\rho_t) 
	+
	\braket{N^{(\omega_t^{(n,m)})}} \mathcal{D}[L_t^{(n,m)} ] (\rho_t) \! \Bigr\} ,
\end{align}	
where 
$\gamma = 
\frac{\kB T}{\hbar}
8 \pi \alpha (q/\text{e})^2  
\frac{\kB T}{m_q c^2}$.

In the main text, we assume that the photon bath 
is always in equilibrium at temperature $T$.
In this case,
the photons exhibit the standard Bose--Einstein statistics 
$ \braket{N^{(\omega)}} = (e^{\beta \hbar \omega} - 1)^{-1}$,
and
the dynamics reduce to
\begin{align}
	\dot{\rho}_t 
	= \frac{i}{\hbar} [\rho_t, H_{t}] 
	& + \gamma 
	\sum_n \!\!
	\sum_{m < n} \!
	\Bigl( \frac{\hbar \omega_t^{(n,m)} }{ \kB T } \Bigr)^3
	\bigl| \braket{E_t^{(m)} | \frac{x}{ \lambda_\text{th} } | E_t^{(n)} } \bigr|^2 \biggl\{ 
	\frac{ e^{\beta \hbar \omega_t^{(n,m)} }  }{e^{\beta \hbar \omega_t^{(n,m)} } - 1} 
\mathcal{D}[L_t^{(m,n)} ] (\rho_t) 
	+
	\frac{1}{e^{\beta \hbar \omega_t^{(n,m)} } - 1}  \mathcal{D}[L_t^{(n,m)} ] (\rho_t) \! \biggr\}  .
\end{align}	

Notice that, in the presence of the equilibrium photon bath, 
the relative transition rates between the system's instantaneous energy eigenstates  satisfy detailed balance, such that 
\begin{align}
	\frac{r_t^{(n \to m)}}{r_t^{(m \to n)}} = e^{\beta   ( E_t^{(n)} - E_t^{(m)} ) } ~.	
\end{align}

We choose a regime for the quantum nonequilibrium thermodynamics that is, at least plausibly, experimentally accessible.
Notice that coherent quantum dynamics are very fast, 
since 
$\kB / \hbar 
= 131$ GHz/K.
Meanwhile, 
for a single electron:
$q/\text{e}=1$, 
$\kB / (m_\text{e} c^2) = 1.69 \times 10^{-10} /$K,
and 
$\lambda_\text{th} = 74.6$ nm/$\sqrt{T/\text{K}}$.
At $T=300$ K,
the thermal energy is 
$\kB T = 25.9$ meV,
the thermal wavelength of the electron is
$\lambda_\text{th} = 4.31$ nm,
while the coherent dynamics occur on the very fast timescale of 
$\beta \hbar = 25.4$ fs,
and the relaxation dynamics occur at the much slower rate 
on the order of 
$\gamma 
= 9.3 \times 10^{-9} / (\beta \hbar)
= 366$ kHz.

 \subsection{Quantum evolution with separation of timescales}
 \label{sec:QuantumDynamicsOfTimescaleSep}
 
 The factor $\kB / (m_\text{e} c^2) = 1.69 \times 10^{-10} /$K
 (in the dissipative $\gamma$ term)
 induces a separation of timescales
 between the coherent dynamics and the consequently much slower relaxation dynamics.
 In particular if, over the timescale $\beta \hbar$,
 the Hamiltonian is approximately constant
 and the approximately-constant relaxation rates are very small
(such that $\beta \hbar r^{(m \to n)} \ll 1$),
then it is useful to 
derive and employ the following
discrete-time dynamics.

It will be fruitful to write the full Lindblad superoperator of Eq.~\eqref{eq:Generic_Lindbladian}
as 
\begin{align}
	\mathcal{L} = C + D ~,
\end{align}
and explore the discrete-time evolution superoperator
$e^{t \mathcal{L} } $
over a duration $t$ where $\mathcal{L}$ is approximately
unchanging,
and in a regime where $(tD)^2$ is negligible.
In our case, $C(\cdot) = \frac{i}{\hbar} [\cdot, H]$ is the typical coherent superoperator,
and $D$ describes the dissipative dynamics.
Employing 
the Lie product formula
and noting that 
$e^{t D / N } = I + t D / N + \mathcal{O}\bigl( (t D)^2 \bigr)$,
we find
\begin{align}
	e^{t \mathcal{L} }
	&= 
	e^{t (C+D) }
	=
	\lim_{N \to \infty} \bigl( e^{t C /N} e^{t D/N} \bigr)^N \\
	&\approx e^{tC}
	    + \lim_{N \to \infty} \sum_{n=0}^{N-1} e^{ntC/N} \tfrac{t}{N} D e^{(N-n-1)t C/N} \\
	&=e^{tC} + \int_0^t dt' \, e^{t' C} D e^{(t-t') C}    
	 ~.
	 \label{eq:Perturbative_dissipation_ev}
\end{align}

Using Eq.~\eqref{eq:Perturbative_dissipation_ev}
together with  Eq.~\eqref{eq:Generic_Lindbladian} explicitly,
a long calculation yields
\begin{align}
	e^{t \mathcal{L} }(
	\rho)
	&= 
	U \rho U^\dagger - \tfrac{1}{2} t \sum_n \bigl[ I - \ket{n} \! \bra{n} \tr \bigr] \Bigl( \bigl\{ \sum_{m \neq n} r^{(m \to n)} \ket{m} \! \bra{m} , \, U \rho U^\dagger \bigr\} \Bigr)
	~,
	\label{eq:Final_discrete_dyn}
\end{align}
where 
$U = e^{-i H t / \hbar}$
and $\{ \cdot , \cdot \}$ is the anticommutator.

In our simulations, we restrict the dynamics at each time-step to the twenty
lowest-lying instantaneous energy eigenstates, so that $m$ and $n$ in Eq.~\eqref{eq:Final_discrete_dyn} both range over these twenty states.
To approximate the dynamics, we compound
a sequence of many discrete steps---each of the form 
$e^{(\delta t) \mathcal{L}_t}$
of Eq.~\eqref{eq:Final_discrete_dyn},
but with each step using the instantaneous Lindblad
superoperator $\mathcal{L}_t$ 
induced by 
(i) 
the instantaneous energy eigenstates
of the Hamiltonian $H_t$
and 
(ii)
the
instantaneous transition rates
given by
Eqs.~\eqref{eq:AbsorptionRate_final} and \eqref{eq:EmmissionRate_final}.

Recall that our simulations assume room temperature (300 K), 
and so 
transition rates are 
on the order of 
$\gamma 
= 9.3 \times 10^{-9} / (\beta \hbar)
= 366$ kHz.
Accordingly, we take 
$\delta t = 4 \times 10^{5} \beta \hbar \approx 10$ ns
as the duration of each discrete time-step.~\footnote{This regime of $(\gamma \delta t )^2 \ll 1$ indeed satisfies the assumption of negligible $(\delta t D)^2$ that justifies the use of Eq.~\eqref{eq:Final_discrete_dyn} in our simulations.}
Each step of the discrete dynamic thus allows for many coherent oscillations with only perturbative decoherence.	
 We choose the total duration of the protocol to be 
 $\tau = 4 \times 10^9 \beta \hbar 
 \approx  100 \, \mu$s,
 which is long enough to allow for many 
 transitions,  yet short enough to keep the system away from equilibrium throughout the protocol.

\subsection{Work with separation of timescales}
\label{sec:QuantumThermoDynamicsOfTimescaleSep}

Recall that the expectation value of work 
can be calculated as 
$	
\braket{W}_{\rho_0} 
= 
\int_0^\tau \tr \bigl( \rho_t \dot{H}_t  \bigr) \, dt
$.
During a simulation, the integral is approximated numerically 
via breaking time into many bins, each much smaller than $\tau$.
In typical simulations, with $N$ timesteps of duration $\tau' = \tau/N$,
the expected work can be approximated as~\cite{Croo98a}
\begin{align}
	\braket{W}_{\rho_0} 
	\stackrel{?}{\approx}
	\sum_{n=0}^{N-1}
	\tr \bigl[ \rho_{n \tau'} ( H_{(n+1) \tau'} - H_{n \tau'} ) \bigr] ~,
	\label{eq:TypicalNumericalWorkEstimator}
\end{align}	
which implicitly assumes that the density matrix 
does not change appreciably during the timestep.

However, because of the separation of timescales in our simulations,
the density matrix goes through many coherent oscillations (in the relative phase between energy eigenstates) during a single timestep.
Accordingly, 
treating the density matrix as a constant during each timestep,
as done in Eq.~\eqref{eq:TypicalNumericalWorkEstimator}
would not be appropriate in our situation.

Rather---with $N$ timesteps of duration $\tau' = \tau/N$, 
and in a regime where the rate of change of the Hamiltonian is approximately constant throughout each timestep---the expected work is well approximated by 
\begin{align}
	\braket{W}_{\rho_0} 
	&= 
	\int_0^\tau \tr \bigl( \rho_t \dot{H}_t  \bigr) \, dt  \nonumber \\
	& =
	\sum_{n=0}^{N-1}
	  	\int_{n \tau'}^{(n+1) \tau'} \tr \bigl( \rho_t \dot{H}_t  \bigr) \, dt \\
	 &\approx  	
	 \sum_{n=0}^{N-1}
	   \tr \biggl[ \Bigl( \int_{n \tau'}^{(n+1) \tau'} \rho_{n \tau'} \, dt \Bigr) \frac{ H_{(n+1) \tau'} - H_{n \tau'} }{\tau'}  \biggr] \\
	&=  	
	\sum_{n=0}^{N-1}
	\tr \Bigl[ \, \overline{\rho_{n \tau'}} \, \bigl(  H_{(n+1) \tau'} - H_{n \tau'}  \bigr) \Bigr] 
	~,
	\label{eq:BetterWorkApprox}
\end{align}	
where 
\begin{align}
\overline{\rho_{n \tau'}} \coloneqq 
\frac{1}{\tau'} \int_{n \tau'}^{(n+1) \tau'} \rho_t \, dt
\end{align}
is the time-averaged density matrix, 
averaged over the duration of the $n^\text{th}$
timestep. (Recall that our simulations use a timestep of $4 \times 10^5 \beta \hbar$.)
It is easy to see how Eq.~\eqref{eq:BetterWorkApprox} would reduce to Eq.~\eqref{eq:TypicalNumericalWorkEstimator}
in a regime where $\rho_t$ is approximately constant throughout the duration of a timestep.

As a first-order approximation to the time-averaged density matrix within a single timestep,
we can assume that the system does not make any dissipative transition.
In this case, $e^{t \mathcal{L}} = e^{t (C+D)} \approx e^{tC} $,
and the time-averaged density matrix will be well approximated by:
\begin{align}
	\overline{\rho_{n \tau'}} 
	&\approx 
	\frac{1}{\tau'} \int_{n \tau'}^{(n+1) \tau'} e^{(t - n \tau')C_{n \tau'}} (\rho_{n \tau'}) \, dt \\
	&=
	\frac{1}{\tau'} \int_{0}^{\tau'} e^{t C_{n \tau'}} (\rho_{n \tau'}) \, dt \\
	&= \frac{1}{\tau'} \sum_{E, E' \in V_{H_{n \tau'}}}
	  \ket{E} \! \bra{E} \rho_{n \tau'} \ket{E'} \! \bra{E'} 
	  \int_{0}^{ \tau'} e^{i (E'-E)t/\hbar}  \, dt 
	\\
&=  \sum_{E \in V_{H_{n \tau'}}}
	\ket{E} \! \bra{E} \rho_{n \tau'} \ket{E} \! \bra{E} +
	\sum_{E' \in V_{H_{n \tau'}} \setminus \{ E \} }
	\frac{ \sin(\omega \tau' / 2) }{ \omega \tau' / 2 } \,
	e^{i \omega \tau' / 2} 
	\ket{E} \! \bra{E} \rho_{n \tau'} \ket{E'} \! \bra{E'} 
~,
\end{align}
where we have used the shorthand 
$\omega = (E' - E)/\hbar$.
Notice that energetic coherences get averaged out
via a decaying envelope with magnitude $2/(\omega \tau')$.

\end{document}